\documentclass[12pt]{article}
\usepackage[utf8]{inputenc}
\pdfoutput=1
\usepackage{fullpage}
\usepackage{comment}
\usepackage{amssymb}
\usepackage{amsthm}
\usepackage{amsmath}
\usepackage{mathdots}
\usepackage{hyperref}
\usepackage{xcolor}
\hypersetup{
    colorlinks,
    linkcolor={red!50!black},
    citecolor={blue!50!black},
    urlcolor={blue!80!black}
}
\usepackage{graphicx}
\usepackage{wrapfig}
\usepackage{soul}

\usepackage{blindtext}
\usepackage{titlesec}

\usepackage{array}
\usepackage{makecell}

\def\inline#1:{\par\vskip 3pt\noindent{\bf #1:}\hskip 10pt}

\def\dnsparagraph{\vspace{-7pt}\paragraph}

\newtheorem{theorem}{Theorem}[section]
\newtheorem{lemma}[theorem]{Lemma}
\newtheorem{observation}[theorem]{Observation}

\newtheorem{corollary}[theorem]{Corollary}
\newtheorem{proposition}[theorem]{Proposition}
\newtheorem{claim}[theorem]{Claim}

\newtheorem{definition}[theorem]{Definition}

\newtheorem{remark*}[theorem]{Remark}
\newtheorem{remark}[theorem]{Remark}

\def\SUPPORTED{\mbox{\sf SUPPORTED} }
\def\CONGEST{\mbox{\sf CONGEST} }
\def\LOCAL{\mbox{\sf LOCAL} }
\def\LTIME{T_{\mbox{\sf local}}}
\def\SUPTIME{T_{\mbox{\sf sup}}}
\def\SUPSPACE{S_{\mbox{\sf sup}}}
\def\LCL{\mbox{\sf LCL}}

\def\cP{\mathcal{P}}
\def\cA{\mathcal{A}}
\def\tT{{\check T}}
\def\RT{\mathcal{R}}
\def\RLT{\RT_L}
\def\RCT{\RT_C}
\def\RFT{\RT_F}
\def\layer{layer}
\def\parent{\mathsf{parent}}
\def\AA{\mbox{\sf APPR}}
\def\PROC{\mbox{\sf P}_{up}}
\def\lleaves{lleaves}
\def\main{int}

\def\leaftree{L-tree}
\def\leaftrees{L-trees}
\def\inttree{I-tree}

\def\cuttree{cut I-tree}
\def\cuttrees{cut I-trees}
\def\fulltree{full I-tree}
\def\fulltrees{full I-trees}
\def\Inf{{\sf Inf}}
\def\InS{\mbox{\sf S}}
\def\UN{un}

\def\CDS{\mbox{\sf CDS}}
\def\PCC{\mbox{\sf CC}}

\def\Dopt{D^{\texttt{opt}}}

\def\depth{\texttt{depth}}
\def\wt{\texttt{weight}}
\def\col{\mbox{{\sf color}}}
\def\path{\mbox{{\sf path}}}
\def\cluster{\mbox{\sf cluster}}
\def\out{\mbox{\sf out}}

\usepackage{algorithm}
\usepackage{algpseudocode}
\usepackage{authblk}

\usepackage{tikz}
\usetikzlibrary{arrows, patterns, patterns.meta, arrows.meta, calc, decorations.markings, math, decorations.pathreplacing, calligraphy}
\makeatletter
\tikzset{circle split part fill/.style  args={#1,#2}{%
 alias=tmp@name, %
  postaction={%
    insert path={
     \pgfextra{%
     \pgfpointdiff{\pgfpointanchor{\pgf@node@name}{center}}%
                  {\pgfpointanchor{\pgf@node@name}{east}}%
     \pgfmathsetmacro\insiderad{\pgf@x}
      \fill[#1] (\pgf@node@name.base) ([xshift=-\pgflinewidth]\pgf@node@name.east) arc
                          (0:180:\insiderad-\pgflinewidth)--cycle;
      \fill[#2] (\pgf@node@name.base) ([xshift=\pgflinewidth]\pgf@node@name.west)  arc
                           (180:360:\insiderad-\pgflinewidth)--cycle;            
         }}}}}  
 \makeatother  
\newcommand{\drawgz}[4]{
\node[draw, circle, fill=red!20] (a#4) at (-0.5+#1-0.5, 0+#2-0.5) {1};
\node[draw, circle, fill=red!20] (c#4) at (0.5+#1-0.5, 1.5+#2-0.5) {#3};
\node[draw, minimum size = 1cm, rectangle] (b#4) at (1.5+#1-0.5, 0+#2-0.5) {$K_{\chi - 2}$};
\draw (a#4) -- (b#4);
\draw (c#4) -- (a#4);
\draw (c#4) -- (b#4);
}

\begin{document}

\author[1]{Akanksha Agrawal}
\author[1]{John Augustine}
\author[2]{David Peleg}
\author[1]{Srikkanth Ramachandran}
\affil[1]{Indian Institute of Technology Madras}
\affil[2]{Weizmann Institute of Science, Israel}

\title{Recurrent Problems in the LOCAL Model}

\date{ }

\maketitle

\begin{abstract}
The paper considers the \textsf{SUPPORTED} model of distributed computing introduced by Schmid  and Suomela [HotSDN'13], generalizing the \textsf{LOCAL} and \textsf{CONGEST} models. In this framework, multiple instances of the same problem, differing
from each other by the subnetwork to which they apply, recur over time,
and need to be solved efficiently online. To do that, one may rely on
an initial preprocessing phase for computing some useful information.
This preprocessing phase makes it possible, in some cases, to obtain improved
distributed algorithms, overcoming locality-based time lower bounds. 

A first contribution of the current paper is expanding the spectrum of problem types to which the \textsf{SUPPORTED} model applies. In addition to \emph{subnetwork}-defined recurrent problems, we introduce also recurrent problems of two additional types: (i) instances defined by \emph{partial client sets}, and (ii) instances defined by \emph{partially fixed outputs}. 

Our second contribution is exploring and illustrating the versatility
and applicability of the \textsf{SUPPORTED} framework via examining
new recurrent variants of three classical graph problems. 
The first problem is  \textit{Minimum Client Dominating Set (\textsf{CDS})}, a recurrent version of the classical dominating set problem with each recurrent instance requiring us to dominate a partial client set.
We provide a \textit{constant time approximation scheme} for the \textsf{CDS} problem on trees and planar graphs,
overcoming the $\Omega(\log^*n)$ based locality lower bound.
The second problem is \textit{Color Completion (\textsf{CC})}, a recurrent version of the coloring problem in which each recurrent instance comes with a partially fixed coloring (of some of the vertices) that must be completed.
We study the minimum number of new colors and the minimum total number of colors necessary for completing this task.
We show that it is not possible to find a \textit{constant time approximation scheme} for the minimum number of additional colors required to complete the precoloring. On the positive side, we provide an algorithm that computes a $2$-approximation for the total number of colors used in the completed coloring (including the set of pre-assigned colors), as well as a one round algorithm for color completion that uses an asymptotically optimal number of colors.

The third problem we study is a recurrent version of Locally Checkable Labellings (LCL) on paths of length $n$. We show that such problems have complexities that are either $\Theta(1)$ or $\Theta(n)$, extending the results of
Foerster et al. [INFOCOM'19].
\end{abstract}

\tableofcontents

\newpage

\section{Introduction}

The area of distributed network algorithm concerns the development and analysis of distributed algorithm operating on a network of processors interconnected by communication links. In particular, a substantial body of research has been dedicated to the development of various graph algorithms for problems whose input consists of the network topology. Examples for such problems are finding maximal independent set (MIS) for the network, finding a maximal or maximum matching (MM), a minimum dominating set (MDS), a proper coloring with few colors, and so on, and considerable efforts were invested in developing sophisticated and highly efficient algorithms for these problems. Such algorithms are particularly significant in settings where the distributed network at hand is {\em dynamic}, and its topology keeps changing at a high rate.

The observation motivating the current study is that in many practical settings, the network itself may be static, or change relatively infrequently. In such settings, problems depending solely on the graph structure need be solved only once. In contrast, there are a variety of other problems, related to computational processes that occur repeatedly in the network, which need to be solved at a much higher frequency, and whose input consists of the network topology together with some other (varying) elements. For such problems, the traditional model might not provide a satisfactory solution, in the sense that it may be unnecessarily expensive to solve the entire problem afresh for each instance. Rather, it may be possible to derive improved algorithmic solutions that take advantage of the fact that the network topology is static. 
We refer to such problems as {\em recurrent problems}.

We envision that depending on the desired problems that the network needs to support, one can compute and store additional information about the topology of the network within each node, to enable recurrent problems to be solved faster. Inherently this captures an aspect of network design. When a network is built, it maybe useful to compute useful information about its topology keeping in mind the recurrent problems that it must support during its lifetime.

This framework has already been studied in literature as the $\SUPPORTED$ model \cite{supmodel}, wherein the recurrent problems are simply instances of the original problem but on a (edge induced) subgraph of the original graph. Edges of the original graph remain valid communication links.
We believe that the $\SUPPORTED$ model (as mentioned in \cite{supmodel}) does not fully capture all recurrent problems. To demonstrate this, we study a couple of natural extensions of the classical local problems of coloring and dominating set. 

\subsection{Recurrent Problems}

We consider graph-related optimization problems each of whose {\em instances} $\langle G,\InS\rangle$ consists of a network topology $G=(V,E)$, on which the distributed algorithm is to be run, and some problem-specific input $\InS$. 
The term "recurrent problem" refers to a setting where the network $G$ is fixed, and is the same in all instances (hence we often omit it). Formally, there is a stream of instances 
that arrive from time to time and must be solved efficiently.
The optimization problem itself may be a variant of some classical graph optimization problem, except it has some additional restrictions, specified in each instance $\InS$. 
Two concrete types of restrictions that are of particular interest are 
{\em partial client set (PCS)} and 
{\em partially fixed output (PFO)}.

\dnsparagraph{Partial client set (PCS)}
An instance $\InS$ restricted in this manner specifies a subset $C\subseteq V$ of {\em client vertices} to which the problem applies. The rest of the vertices are not involved (except in their capacity as part of the network).
For example, consider the maximum matching problem. In the PCS variant of this problem, a PCS-restricted instance will specify a vertex subset $C$ such that the matching is only allowed (and required) to connect vertex pairs of $C$.

\dnsparagraph{Partially fixed output (PFO)}
An instance $\InS$ restricted in this manner specifies a part of the output. The rest of the output must be determined by the algorithm.
For example, consider $k$-centers problem (where the goal is to select a subset $C$ of $k$ vertices serving as centers, so as to minimize the maximum distance from any vertex of $V$ to $C$). In the PFO variant of $k$-centers problem, a PFO-restricted instance will specify a vertex subset $C_{pre}$ of $k'$ vertices that were already pre-selected as centers, and the task left to the algorithm is to select the remaining $k-k'$ centers.

Naturally, some recurrent problems may involve other novel restrictions as well as hybrids, thereby opening up the possibility for rich theory to be developed. 

\subsubsection{Two representative examples: \texorpdfstring{$\CDS$}{CDS} and \texorpdfstring{$\PCC$}{PCC}}

In this paper, we will focus on two concrete examples for recurrent problems of practical significance, and use them for illustration.  The first of these two example problems, named $\CDS$,
serves to illustrate a recurrent problem with PCS-restricted instances (where the set of clients changes in each instance). The second problem, named $\PCC$, illustrates a recurrent problem with PFO-restricted instances (where parts of the output are fixed in advance in each instance).

\dnsparagraph{Minimum client-dominating set ($\CDS$)}
In certain contexts, a dominating set $D$ in a network $G$ (i.e., such that every vertex $v\in V$ either belongs to $D$ or has a neighbor in $D$) is used for placing {\em servers} providing some service to all the vertices in the network (interpreted as {\em clients}), in settings where it is required that each vertex is served by a server located either locally or at one of its neighbors. The {\em minimum dominating set (MDS)} problem requires finding the smallest possible dominating set for $G$.

We consider the recurrent variant of the $\CDS$ problem with PCS-restricted instances. This problem arises in settings where the set of clients in need of service does not include all the vertices of $G$, but rather varies from one instance to the next. In such settings, the network $G$ is static, and from time to time, a set of clients $C\subseteq V$, formed in an ad-hoc manner due to incoming user requests, requests to select and establish a (preferably small) subset $D$ of vertices from among their neighbors, which will provide them some service. In other words, the set $D$ is required to dominate the vertices in $C$. On the face of it, solving the minimum dominating set problem once on $G$ may be useful, but not guarantee optimal results for each recurrent instance $\InS$; rather, for each instance $\InS$, it may be necessary to solve the specialized problem once the request appears in the network. Hereafter, we refer to this problem as {\em minimum client-dominating set ($\CDS$)}.

Note that one may also consider a generalized problem that includes also a PFO component, by specifying in each instance $\InS$ also a partial set $D'$ of vertices that were pre-selected as servers (or dominators). Our results are presented for the $\CDS$ problem (without PFO restrictions), but it should be clear that they can easily be extended to the generalized problem with PFO restrictions\footnote{Essentially, for this problem, the pre-selected vertices of $D'$ can be used to satisfy all the clients that neighbor them, leaving us with a smaller set $C'$ of unsatisfied clients that need to be covered.}.

\dnsparagraph{Color Completion ($\PCC$)}
In certain contexts, a proper coloring of a distributed network is used for purposes of scheduling various mutually exclusive tasks over the processors of the network. For example, suppose that performing a certain task by a processor requires it to temporarily lock all its adjacent links for its exclusive use, preventing their use by the processor at the other end. Employing a proper coloring as the schedule (in which all the processors colored by color $t$ operate simultaneously at round $t$) enables such mutually exclusive operation. Naturally, it is desirable to use as few colors as possible, in order to maximize parallelism. 

We consider the recurrent variant of the coloring problem with PFO-restricted instances. 
From time to time we may receive a partial (collision-free) coloring assignment to some subset $C\subseteq V$ of the vertices, representing processors constrained to operate on some particular time slots. We are required to color all the remaining vertices in $V\setminus C$ properly and consistently with the initial coloring. Typically, using colors already occurring in the precoloring (i.e., used by some vertices in the set $C$) is fine, since these time slots are already committed for the task at hand. However, it is desirable to use as few new time slots (or new colors), to minimize the overall time spent on the task. 

Note that one may also consider a generalized problem that includes also a PCS component, by specifying in each instance $\InS$ also a partial set $V'$ of vertices that are interested in being scheduled, and hence need to be colored. Our results are presented for the $\PCC$ problem (without PCS restrictions), but it should be clear that they can easily be extended to the generalized problem with PCS restrictions\footnote{Essentially, for this problem, the vertices of $V\setminus V'$, which do not require coloring, can simply avoid participating in the coloring process.}.

\subsection{The \SUPPORTED model}
The \SUPPORTED model is an extension of the well studied \LOCAL and \CONGEST models with an additional preprocessing phase.
Specifically the solution to a problem in the $\SUPPORTED$ model consists of two stages, (i) a {\em preprocessing } stage and (ii) an {\em online} stage.

\begin{itemize}
\item 
In the preprocessing stage, run an algorithm $\cA_{pre}(G)$ on the topology of the network $G$ and obtain information $\Inf(G)$ to be stored at the network vertices (different vertices may of course store different information).
\item
During runtime, a stream of instances 
arrive. Whenever a new instance $\InS$ arrives, run an algorithm $\cA(\InS,\Inf(G))$ to solve this problem instance. 
\end{itemize}

In view of the fact that the preprocessing stage takes place only once, the particulars of the preprocessing algorithm are less important to us, and we allow it to be arbitrary (even oracular).
For the scope of this paper, in the upper bounds that we show, all our preprocessing phases are decidable, whereas the lower bounds hold for any arbitrary preprocessing.

In the online stage, we insist that the computations performed by each node in a single round must be polynomial in the size of the graph. Therefore even knowledge of the complete network for each node might not be sufficient, as underlying information about the topology (such as chromatic number) might not be computable in polynomial time. 

For a given problem $\Pi$ on a graph $G$, one may seek to optimize several parameters. For the scope of this paper, we consider only two, (i) the round complexity of the online algorithm, i.e., the number of synchronous rounds required to solve each recurrent instance and (ii) the size of the output to each node in the preprocessing phase, i.e., the amount of additional information that needs to be stored in each node of the graph from the preprocessing phase.
We use $\SUPTIME(\Pi, G)$ to denote the worst case online round complexity for any deterministic algorithm across all instances of $\Pi$. We use $\SUPSPACE(\Pi, G)$ to be the optimal size of the output to each node in the preprocessing phase that enables $\Pi$ to be solved in $\SUPTIME(\Pi, G)$ rounds in the online stage. 
We use $\LTIME(\Pi, p)$ to denote the worst case round complexity for $\Pi$ in the classical local model on all graphs with given parameter $p$. Depending on the problem, $\LTIME(\Pi)$ may be described by a combination of different parameters of the input graph, such as the number of nodes $n$ or maximum degree $\Delta$. 

\subsection{Our Contributions}

In Section 2, study the $\CDS$ problem. We first show that even on a path, it is not possible to optimally solve $\CDS$ in $o(n)$ time. We next look at $1 + \epsilon$ approximations. We show that on trees and planar graphs, one can obtain a $1 + \epsilon$ approximation in $O(\frac{1}{\epsilon})$ and $\tilde{O}\left(\frac{1}{\epsilon}^{\log_{24/23}{3}}\right) $ rounds respectively. To achieve these bounds, we only require to store $O(1)$ bits per node as the output of the preprocessing phase.

In Section 3, we study the $\PCC$ problem. We provide an algorithm to complete a given coloring using at most $\chi (\Delta + 1) / k$ new colors in $k$ rounds. We show that for $k = 1$, the number of colors used is asymptotically tight in the worst case.

In Section 4, we study a generic class of problems called Locally Checkable Labellings (LCL). We show that on a path, every LCL problem either has worst case complexity $\Theta(1)$ or $\Theta(n)$. In the specific case of recurrent problems where the online instances are a specific LCL on a sub-path of the given path (as considered in prior works such as \cite{foerster2019power}), we provide an efficient centralized algorithm to classify the LCL into one of the two cases and also construct the distributed algorithm to solve an LCL given its description, thereby extending the results in \cite{foerster2019power}. In our construction, the preprocessing phase requires only $O(1)$ additional bits to be stored per node. 

Finally in Section 5, we provide some partial results on sub-graph maximal matching and sub-graph maximal independent set that could potentially be useful in finding optimal solutions for these problems in the $\SUPPORTED$ model.

\subsection{Related Work}

The \SUPPORTED model for first proposed by Schmid and Suomela \cite{supmodel}. Foerster et al. \cite{foerster2019power} provide several results including lower bounds for problems such as sinkless orientation and approximating independent set. For global network optimization problems, such as minimum spanning tree, near optimal \textit{universal} lower and upper bounds have been shown by (\cite{haeupler2021universally}). 
We stress that in all related prior work above, the problems to be solved are same as the traditional problems, but on a subgraph of the given graph. Most of our solutions here are adaptations of existing algorithms for the relevant problems in the \LOCAL model. 

\textbf{Dominating Set.} Czygrinow et al \cite{CHW08} provided an $O_{\epsilon}(\log^* n)$ round algorithm for a $1 + \epsilon$ approximation for the dominating set problem and it was later extended to bounded genus graphs by Amiri et al\cite{ASS19}. Foerster et al. \cite{foerster2019power} briefly discuss about extending these results to the \SUPPORTED model.

\textbf{Coloring.} Color Completion has been one of the methods used for $\Delta + 1$ coloring graphs in $\log^*n + f(\Delta)$ rounds. Existing algorithms decide on a coloring for a subgraph of the given graph and then recursively complete the chosen coloring. Barenboim \cite{barenboim2016deterministic} provided the first sublinear in $\Delta$ algorithm. The current best known algorithm has round complexity $\log^*n + O(\sqrt{\Delta \log \Delta})$ (see \cite{maus2020local, BEG13, fraigniaud2016local}). Maus \cite{maus2020local} also provided a smooth tradeoff between the number of colors and the round complexity, specifically in $k + \log^* n$ rounds, graphs can be properly colored using $O(\Delta^2 / k)$ colors for any $1 \leq k \leq \sqrt{\Delta}$. We note that Maus's algorithm does not provide a $\Delta + 1$ coloring but rather an $O(\Delta)$ coloring.

\textbf{LCL.} Locally Checkable Labellings (LCL) were first proposed by Naor and Stockmeyer \cite{DBLP:journals/siamcomp/NaorS95}. Chang et al. \cite{chang2019exponential} showed gaps in the deterministic complexity of LCL's. They showed that the worst case deterministic round complexity of LCL's on any hereditary graph class is either $\omega(\log_{\Delta} n)$ or $O(\log^* n)$. They also show that for paths, there is no LCL with complexity $o(n)$ and $\omega(\log^* n)$. Later Chang et al \cite{chang2019time} showed that on trees, the deterministic worst case complexities for LCL's is either $\Theta(1), \Theta(\log^* n), \Theta(\log n)$ or $n^{\Theta(1)}$. They also provide examples of LCL's with complexity $\Theta(n^{1/k})$ for any integer $k$. More recently, Balliu et al. \cite{balliu2021locally} showed that for a more restricted class of LCL problems called homogenous LCL problems, on rooted trees, there is a centralized algorithm that takes as input the description of the LCL and decides which of the above complexity classes it belongs to. Given the LCL, deciding its distributed complexity class on trees was shown to be \textsf{EXPTIME} hard by Chang \cite{chang2020complexity}.

\section{Dominating Sets}

\subsection{Client Dominating Set}

\begin{definition}[Client Dominating Set]
    Given a graph $G$ and a subset of its vertices $C \subseteq V(G)$, called the client set, we say that a subset $D$ is a client dominating set of $G, C$ if for every client $c \in C$, there exists $v \in D$ such that either $v = c$ or $v$ is a neighbor of $c$. 
\end{definition}

\begin{definition}[Minimum Client Dominating Set (\CDS)]
    Given a graph $G$ and a subset of its vertices $C \subseteq V(G)$, called the client set, find a client dominating set of minimum size.
\end{definition}

The $\CDS$ problem is of course a generalization of the Dominating Set problem as the dominating set is precisely the case when $C = V(G)$. It is also possible to reduce the $\CDS$ problem to an instance of a Dominating Set problem. Given a graph $G$ and a client set $C$, we can construct a graph $G_C$ which is obtained by adding a path on two vertices, $P_2$ to $G$ and connecting every nonclient vertex (i.e. $V(G) \setminus C$) to one end of the path $P_2$. See Figure \ref{fig:cds-ds:redn} (a). 

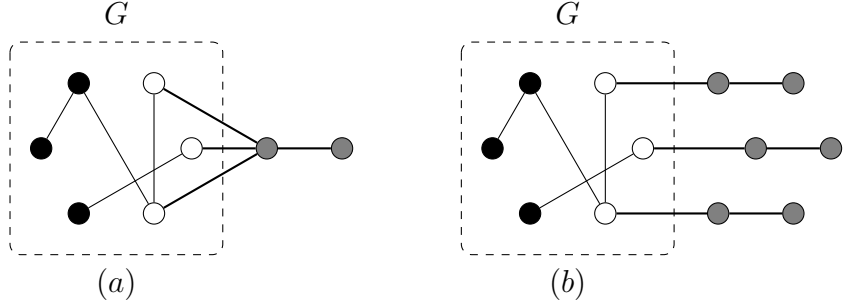
\begin{figure}
    \centering
    \begin{tikzpicture}
        \foreach \i in {2,...,4} {
            \node[draw, circle, inner sep=1mm, fill=black] (A\i) at (\i*60:1) {};
        }
        \foreach \i in {5,...,7} {
            \node[draw, circle, inner sep=1mm] (A\i) at (\i*60:1) {};
        }
        \draw (A2) -- (A3);
        \draw (A2) -- (A5);
        \draw (A7) -- (A5);
        \draw (A4) -- (A6);
        
        \node[draw, circle, inner sep=1mm, fill=gray] (B1) at (0:2) {};
        \node[draw, circle, inner sep=1mm, fill=gray] (B2) at (0:3) {};
        \draw[thick] (B1) -- (B2);
        
        \draw[thick] (A5) -- (B1);
        \draw[thick] (A6) -- (B1);
        \draw[thick] (A7) -- (B1);
        
        \draw[rounded corners, dashed] (135:2) rectangle (315:2) {};
        \node at (90:1.8) {$G$};
        
        \foreach \i in {2,...,4} {
            \node[draw, circle, inner sep=1mm, fill=black] (C\i) at ($(6,0) + (\i*60:1)$) {};
        }
        \foreach \i in {5,...,7} {
            \node[draw, circle, inner sep=1mm] (C\i) at ($(6,0) + (\i*60:1)$) {};
        }
        \draw (C2) -- (C3);
        \draw (C2) -- (C5);
        \draw (C7) -- (C5);
        \draw (C4) -- (C6);
        \draw[rounded corners, dashed] ($(6, 0) + (135:2)$) rectangle ($(6,0) + (315:2)$) {};
        \node at ($(6,0) + (90:1.8)$) {$G$};
        
        \node[draw, circle, inner sep=1mm, fill=gray] (D1) at ($(7.5,0) + (5*60:1)$) {};
        \node[draw, circle, inner sep=1mm, fill=gray] (D2) at ($(8.5,0) + (5*60:1)$) {};

        \node[draw, circle, inner sep=1mm, fill=gray] (E1) at ($(7.5,0) + (6*60:1)$) {};
        \node[draw, circle, inner sep=1mm, fill=gray] (E2) at ($(8.5,0) + (6*60:1)$) {};

        \node[draw, circle, inner sep=1mm, fill=gray] (F1) at ($(7.5,0) + (7*60:1)$) {};
        \node[draw, circle, inner sep=1mm, fill=gray] (F2) at ($(8.5,0) + (7*60:1)$) {};
        
        \draw[thick] (D1) -- (D2);
        \draw[thick] (E1) -- (E2);
        \draw[thick] (F1) -- (F2);
        \draw[thick] (C5) -- (D1);
        \draw[thick] (C6) -- (E1);
        \draw[thick] (C7) -- (F1);
        
        \node at (0, -1.8) {$(a)$};
        \node at (6, -1.8) {$(b)$};
        
    \end{tikzpicture}
    \caption{$(a)$ PTAS preserving reduction $(b)$ Locality preserving reduction, black vertices are clients, thick edges and gray vertices are added.}
    \label{fig:cds-ds:redn}
\end{figure}

\begin{claim}\label{clm:cds-ds:redn}
     Given a graph $G$ and a client set $C \subseteq V(G)$, consider the graph $G_C$ with
     \begin{itemize}
         \item $V(G_C) = V(G) \cup \{u_1, u_2\}$ where $u_1, u_2 \not\in V(G)$ are two new vertices
         \item $E(G_C) = E(G) \cup \{(u_1, v) \mid v  \in V(G) \setminus C\} \cup \{(u_1, u_2)\}$
     \end{itemize}
     For any $D \subseteq V(G_C)$, 
     $D \cap V(G)$ is a client dominating set of $G, C$ if and only if $D \cup \{u_1\}$ is a dominating set of $G_C$.
\end{claim}

\begin{proof}
    ($\Rightarrow$) Suppose $D \cap V(G)$ is a client dominating set of $G, C$, then all vertices in $C$ have a neighbor in $D \cap V(G)$. Now we look at those vertices in $G_C$ that are dominated by $D \cap V(G)$. The only possible vertices that are not dominated in $G_C$ are the non clients $V(G) \setminus C$ and the two vertices $u_1, u_2$. Notice that $u_1$ dominates all of them. Therefore $(D \cap V(G)) \cup \{u_1\}$ dominates $G_C$. $D \cup \{u_1\}$ is the almost the same set, except possibly with $u_2$ removed. As $u_2$ is not necessary when $u_1$ is present, $D \cup \{u_1\}$ must dominate $G_C$.
    
    ($\Leftarrow$) Suppose $D \cup \{u_1\}$ is a dominating set for $G_C$. $u_1$ only dominates the vertices $V(G) \setminus C, u_1, u_2$. The dominators of the remaining vertices (i.e $C$) must thus be present solely in $V(G)$, i.e., they must be $(D \cup \{u_1\})\cap V(G) = D \cap V(G)$.
\end{proof}

Notice that given a dominating set $D$ of $G_C$, one can replace $u_2$ (if it exists in the solution) with $u_1$ and then by Claim \ref{clm:cds-ds:redn}, $D \cap V(G)$ is a client dominating set. If $D$ is optimal, then $D \cap V(G)$ must an optimal client dominating set. Furthermore, suppose a $1 + \epsilon$ approximation for the dominating set is known for $G_C$, then using Claim \ref{clm:cds-ds:redn} we can get a dominating set of size $(1 + \epsilon)(|D^*| + 1) - 1 = (1 + \epsilon)|D^*| + \epsilon \leq (1 + 2\epsilon) |D^*|$. 

The above reduction holds only for centralized algorithms. Since the above reduction does not preserve locality, non-clients which are far apart in $G$ may be close in $G_C$, a distributed algorithm for dominating set does not immediately imply a distributed algorithm for $\CDS$ with the same round complexity. While we are unable to provide a locality preserving reduction for $1 + \epsilon$ approximating a dominating set, we shall discuss one attempt, which is a slight modification of the above. To each non-client, connect a different path of length $2$, instead of the same path as we have done here (See Figure \ref{fig:cds-ds:redn} (b)). While the new reduction is locality preserving and one can obtain an optimal solution via the new reduction, it does not seem straightforward to obtain an approximation. The reason is that the size of the dominating set for $G_C$ is more than the corresponding client dominating set by an additive $|V(G)| - |C|$ term. Thus if $D^*$ is a client dominating set, the corresponding dominating set in $G_C$ has size,  $(1 + \epsilon) (|D^*| + |V(G)| - |C|) - (|V(G)| - |C|) = (1 + \epsilon) |D^*| + \epsilon (|V(G)| - |C|)$. The additive term $\epsilon (|V(G)| - |C|)$ is too expensive and does not lead to even a constant approximation as $|D^*|$ could be arbitrarily small (even $1$) compared to $\epsilon (|V(G)| - |C|)$. 

\subsection{Lower Bound for Paths}\label{ss:lb}

We establish two lower bounds for $\CDS$ on a path.
First, we argue that, regardless of the preprocessing, the online runtime of
every (exact) deterministic distributed algorithm for the $\CDS$ problem
must take time $\Omega(D)$ on networks of diameter $D$.
Second, we show that the online runtime of every deterministic distributed approximation algorithm for $\CDS$ with ratio $1+\epsilon$ must require time $\Omega(1/\epsilon)$ on some input.

\begin{theorem}
\label{thm:ctas:path:lb}
    Let $\cA$ be a deterministic distributed local algorithm for $\CDS$ with arbitrary preprocessing. Then there exists some input for which $\cA$ requires $\Omega(D)$ time.
\end{theorem}

\begin{proof}
We prove the statement by contradiction. Suppose there exists a deterministic algorithm $\cA$ whose worst case run time is $o(D)$.
Consider a path $P=(v_1, v_2, \dots v_n)$ where $n={4k+2}$ 
for even $k$ 
and the following two instances of clients (see Figure \ref{fig:ctas:path:lb}):
\begin{enumerate}
        \item $C_1 = \{v_2, v_4, \dots v_{4k}\}$, i.e., every vertex at an odd distance from the leftmost vertex except $v_n$.
        \item $C_2 = \{v_4, v_6, \dots v_{4k+2}\}$, i.e., every vertex at an odd distance from the leftmost vertex except $v_{2}$.
\end{enumerate}
    
\begin{figure}[htb]
    \centering
    \begin{tikzpicture}
        \node[draw, circle, inner sep = 6pt] (n1) at (1, 0) {};
        \node[draw=red, circle, inner sep = 6pt] (n2) at (2, 0) {};
        \node (text0) at (2.0, 0) {$v_2$};
        \node[draw=red, circle, inner sep = 7pt] (nn2) at (2, 0) {};
        \node[draw, circle, inner sep = 6pt] (n3) at (3, 0) {};
        \node (text0) at (14.0, 0) {$v_n$};
         \node[draw=red, circle, inner sep = 6pt] (n4) at (4, 0) {};
        \node[draw=red, circle, inner sep = 7pt] (nn4) at (4, 0) {};
        \draw (n1) -- (nn2);
        \draw (nn2) -- (n3);
        \draw (n3) -- (nn4);

        \draw[decoration={brace,mirror,raise=15pt},decorate] (1-0.2,0) -- node[below=20pt] 
        {$2k$} (5+0.3, 0);
        
        \draw[decoration={brace,mirror,raise=15pt},decorate] (6-0.2,0) -- node[below=20pt] 
        {$2k$} (12+0.2, 0);
        
        \node at (5-0.3, 0)[circle,fill,inner sep=0.5pt]{};
        \node at (5, 0)[circle,fill,inner sep=0.5pt]{};
         \node (text0) at (2.0, -3) {$v_2$};
         \node at (5+0.3, 0)[circle,fill,inner sep=0.5pt]{};
        
        \node[draw, circle, inner sep = 6pt] (n5) at (6, 0) {};
        \node[draw=red, circle, inner sep = 6pt] (n6) at (7, 0) {};
        \node[draw=red, circle, inner sep = 7pt] (nn6) at (7, 0) {};
        \node[draw=blue, circle, inner sep = 6pt, line width = 1pt] (n7) at (8, 0) {};
        \node[draw=red, circle, inner sep = 6pt] (n8) at (9, 0) {};
        \node[draw=red, circle, inner sep = 7pt] (nn8) at (9, 0) {};
        \node (text0) at (14.0, -3.0) {$v_n$};
        
        \draw (n5) -- (nn6);
        \draw (nn6) -- (n7);
        \draw (n7) -- (nn8);
        
        \node at (10-0.3, 0)[circle,fill,inner sep=0.5pt]{};
        \node at (10, 0)[circle,fill,inner sep=0.5pt]{};
        \node at (10+0.3, 0)[circle,fill,inner sep=0.5pt]{};
        
        \node[draw, circle, inner sep = 6pt] (n9) at (11, 0) {};
        \node[draw=red, circle, inner sep = 6pt] (n10) at (12, 0) {};
        \node[draw=red, circle, inner sep = 7pt] (nn10) at (12, 0) {};
        \node[draw, circle, inner sep = 6pt] (n11) at (13, 0) {};
        \node[draw, circle, inner sep = 6pt] (n12) at (14, 0) {};
        
        \draw (n9) -- (nn10);
        \draw (nn10) -- (n11);
        \draw (n11) -- (n12);
        
        \node at (7, -1.5) {(a) Instance $C_1$};

        \node[draw, circle, inner sep = 6pt] (m1) at (1, -3) {};
        \node[draw, circle, inner sep = 6pt] (mm2) at (2, -3) {};
        \node[draw, circle, inner sep = 6pt] (m3) at (3, -3) {};
        \node[draw=red, circle, inner sep = 6pt] (m4) at (4, -3) {};
        \node[draw=red, circle, inner sep = 7pt] (mm4) at (4, -3) {};
        \draw (m1) -- (mm2);
        \draw (mm2) -- (m3);
        \draw (m3) -- (mm4);
        
        \node at (5-0.3, -3)[circle,fill,inner sep=0.5pt]{};
        \node at (5, -3)[circle,fill,inner sep=0.5pt]{};
        \node at (5+0.3, -3)[circle,fill,inner sep=0.5pt]{};
        
        \node[draw, circle, inner sep = 6pt] (m5) at (6, -3) {};
        \node[draw=red, circle, inner sep = 6pt] (m6) at (7, -3) {};
        \node[draw=red, circle, inner sep = 7pt] (mm6) at (7, -3) {};
        \node[draw=blue, circle, inner sep = 6pt, line width = 1pt] (m7) at (8, -3) {};
        \node[draw=red, circle, inner sep = 6pt] (m8) at (9, -3) {};
        \node[draw=red, circle, inner sep = 7pt] (mm8) at (9, -3) {};
        
        \draw (m5) -- (mm6);
        \draw (mm6) -- (m7);
        \draw (m7) -- (mm8);
        
        \node at (10-0.3, -3)[circle,fill,inner sep=0.5pt]{};
        \node at (10, -3)[circle,fill,inner sep=0.5pt]{};
        \node at (10+0.3, -3)[circle,fill,inner sep=0.5pt]{};
        
        \node[draw, circle, inner sep = 6pt] (m9) at (11, -3) {};
        \node[draw=red, circle, inner sep = 6pt] (m10) at (12, -3) {};
        \node[draw=red, circle, inner sep = 7pt] (mm10) at (12, -3) {};
        \node[draw, circle, inner sep = 6pt] (m11) at (13, -3) {};
        \node[draw=red, circle, inner sep = 6pt] (m12) at (14, -3) {};
        \node[draw=red, circle, inner sep = 7pt] (mm12) at (14, -3) {};
        
        \draw (m9) -- (mm10);
        \draw (mm10) -- (m11);
        \draw (m11) -- (mm12);
        
        \node at (7, -4) {(b) Instance $C_2$};
        
    \end{tikzpicture}
    \caption{The instances $C_1$ and $C_2$, differing
    in $v_2$ and $v_n$.
    Red double circles denote clients. 
    The blue node is $v_{2k+1}$.
    }
    \label{fig:ctas:path:lb}
\end{figure}
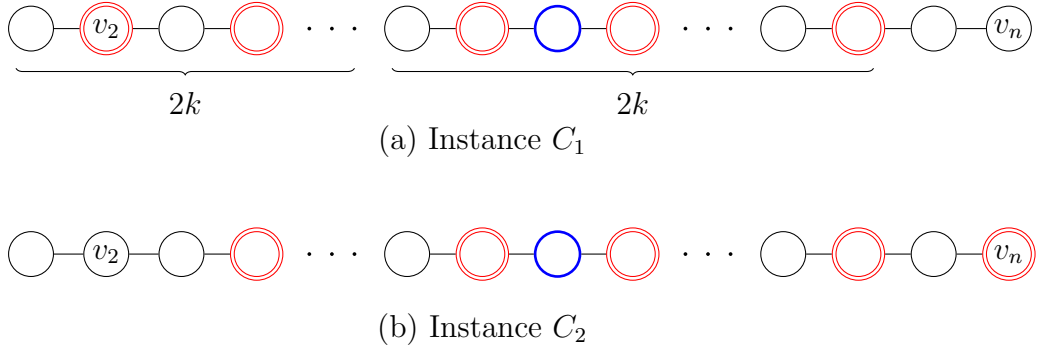

Both these instances have unique optimal solutions that are disjoint. For $C_1$, the optimal solution is to place the dominators at $v_3, v_7, \dots v_{4k-1}$, whereas for $C_2$ the optimal solution is to place them at $v_5, v_9, \dots v_{4k+1}$. Consider the vertex $v_{2k+1}$. It must be chosen as a dominator in exactly one of the two given instances. Since $\cA$ operates in $t = o(D) = o(k)$ rounds, the inputs in the $t$-neighborhood of $v_{2k+1}$, which are observable to $v_{2k+1}$ during the execution, are identical in both instances, and hence the output of $v_{2k+1}$ must be identical as well, yielding the desired contradiction.
\end{proof}

\begin{theorem}
Let $\cA$ be a deterministic distributed local approximation algorithm for $\CDS$,
with arbitrary preprocessing, whose online runtime on every path and every instance is at most $k=4\ell+1$ for some integer $\ell\ge 1$.
There exists a network and a set of clients for which the approximation ratio
of $\cA$ is at least $1+1/(k+2)$.
\end{theorem}

\begin{figure}[htb]
\centering
\includegraphics[width=5in]{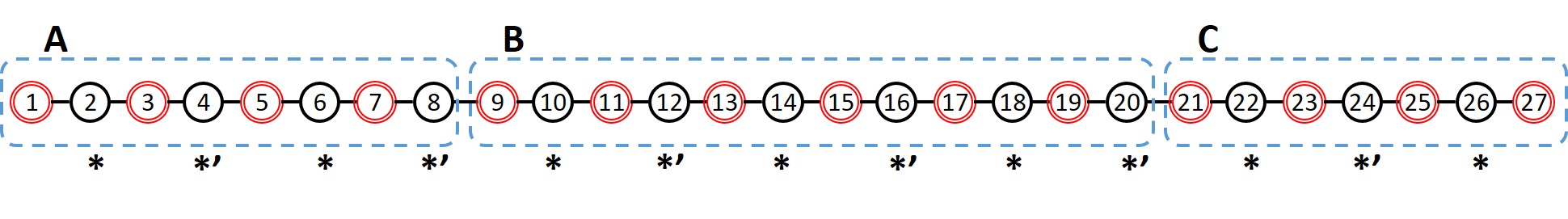}
\caption{An illustration of the instance $(P,S)$ for $\ell=1$. Here $k=5$.
The client vertices of the set $S$ are drawn as double red circles.
The vertices included in the optimal dominating set $D^*$ for $S$
are marked by $*$.
The vertices included in the optimal dominating set $D^{*'}$
for the modified instance $S'$ are marked by $*'$.
}
\label{fig: path P}
\end{figure}

\begin{proof}
Consider an algorithm $\cA$ as specified in the theorem.
Let $P=(v_1,v_2,\ldots,v_{4k+7})$ be a path with $ID(v_i)=i$ for every $i$.
For $i\le j$, denote by $P[v_i,v_j]$ the subpath of a path $P$ from $v_i$ to $v_j$.

Assume an arbitrary preprocessing stage took place, providing the vertices
of $P$ with some additional information.
Let the client set $S$ consist of all the odd-indexed vertices on $P$.
Consider the execution of $\cA$ on $P$ and $S$.
Partition $P$ into three subpaths,
$A=P[v_{1},v_{k+3}]$,
$B=P[v_{k+4},v_{3k+5}]$, and
$C=P[v_{3k+6},v_{4k+7}]$.
(See Fig. \ref{fig: path P} for an illustration.)

Let $D$ be the set of vertices chosen to the dominating set by the algorithm.
For $X\in\{A,B,C\}$,
let $S[X]=S\cap X$ be the set of clients in the subpath $X$, 
and $D[X]=D\cap X$ be the set of dominators selected in $X$.
There are three cases to consider.

\inline Case (1): $|D[B]| \ge 2\ell+2$.
\\
Note that no matter where the dominators of $D[B]$ are placed within
the subpath $B$, at least $\ell+1$ dominators must be selected
in the subpath $C$ in order to dominate all the clients of $S[C]$.
In particular, this holds even if some dominator in $D[B]$
dominates the leftmost client in $C$, $v_{3k+6}$
(node 21 in Fig. \ref{fig: path P}).
Similarly, at least $\ell+1$ dominators must be selected
in the subpath $A$ in order to dominate all the clients of $S[A]$.
(Here, the dominators in $D[B]$ cannot help.)
Altogether, $|D|\ge 4\ell+4$.
On the other hand, note that the unique optimum solution for this instance,
$D^*=\{v_2,v_6,v_{10},\ldots,v_{26}\}$,
consists of $4\ell+3$ dominators (see Fig. \ref{fig: path P}).
Hence in this case, the approximation ratio of $\cA$ is no better than $(4\ell+4)/(4\ell+3)$.

\inline Case (2): $|D[B]| = 2\ell+1$ but $D[B]$ does not dominate all of $S[B]$.
\\
In this case, some of the clients of $S[B]$ must be dominated by dominators
outside the subpath $B$. Inspecting the structure, it is clear that
the only client in $B$ that may be dominated by a dominator outside $B$
is the leftmost client, $v_{k+4}$ (node 9 in Fig. \ref{fig: path P}),
and the only way to do that is by selecting $v_{k+3}$, the rightmost node in $A$,
to $D$.
It is also clear that despite such a selection, $D[A]$ must contain at least
$\ell+1$ {\em additional} dominators in order to dominate all the clients
of $S[A]$.
Also, $|D[C]|\ge \ell+1$ is necessary to dominate $S[C]$.
Hence again, overall $|D|\ge 4\ell+4$, yielding the same approximation ratio
as in case (1).

\inline Case (3): $|D[B]| = 2\ell+1$ and $D[B]$ dominates all of $S[B]$.
\\
Note that in this case, the unique choice is
$D[B]=\{v_{k+5},v_{k+9},\cdot,v_{3k+3}\}$.
Define another instance consisting of
the client set $S'=S\setminus\{v_1,v_{4k+7}\}$, namely, $S$ with the first
and last vertices omitted, and consider the execution of algorithm $\cA$ on this instance.
Notice that in a $k$-round distributed execution, each node is exposed only
to information collected from its distance-$k$ neighborhood.
This implies that the vertices of $B$ see exactly the same view in this new
execution on $S'$ as in the execution on $S$, so their output must be the same.
Hence $D'[B]=D[B]$ (and hence $|D'[B]|=2\ell+1$).
Also note that despite the fact that each of $A$ and $C$ now have
one client fewer than in $S$,
we must have $|D'[A]|\ge \ell+1$ and $|D'[C]|\ge \ell+1$
in order to dominate all the clients of $S'[A]$ and $S'[C]$, respectively.
Hence in total $|D'| \ge 4\ell+3$. However, for this instance the optimum
solution $D^{*'}=\{v_4,v_8,\ldots,v_{4k+4}\}$ is smaller,
consisting of only $k+1=4\ell+2$ vertices (see Fig. \ref{fig: path P}).
Hence in this case, the approximation ratio of $\cA$ is 
$(4\ell+3)/(4\ell+2)$ or higher.

In summary over all cases, the approximation ratio of $\cA$ is
\par
\bigskip
$\mbox{~\hskip 100pt}
\displaystyle \min\left\{\frac{4\ell+3}{4\ell+2}~~,~~ \frac{4\ell+4}{4\ell+3}\right\} ~=~
\frac{k+3}{k+2} ~=~ 1+\frac{1}{k+2}$.
\end{proof}

\subsection{A CTAS for Trees}
\label{sec:MCDS-trees}

In this section we describe the CTAS for $\CDS$ on trees,
prove its correctness and analyze its complexity.

Like the CTAS on a path, the algorithm for trees is based on a preprocessing stage in which the tree is partitioned into subtrees of depth $O(k)$ for integer parameter $k$. Each recurrent instance is then solved by computing a local near-optimal CDS on each subtree, while taking care to ensure that the resulting solutions combine into a $1+4/(k-1)$ approximation of the optimal global solution.
The ``interface'' between adjacent subtrees is more difficult to handle than in the case of paths, as making a single change in the selection in one subtree (e.g., in one of its leaves) might affect several adjacent subtrees, which makes both the algorithm and its analysis somewhat more complex.

Let us first describe the preprocessing stage, which is applied to the network tree $T$.
The algorithm has an integer parameter $\ell\ge 1$ and sets $k=4\ell+1$.
Root the tree $T$ at a root vertex $r_0$, and mark each vertex $v$ by
its {\em layer} $\layer(v)$, namely, its distance from $r_0$
(by definition $\layer(r_0)=0$).
Partition the tree $T$ into subtrees by taking every vertex $v$
with $\layer(v)=pk$ for integer $p\ge 0$ as a root and defining $T[v]$ as 
the subtree of depth $k$ rooted at $v$. See Fig. \ref{fig:sub+cut}(a).
For notational convenience, we sometimes use $T[v]$ to denote also the
{\em vertex set} of the subtree $T[v]$.
Also, for any subtree $T[v]$ and vertex set $X\subseteq T$,
we denote $X[v] = X \cap T[v]$.

\begin{figure}[htb]
\centering
\includegraphics[width=5.9in]{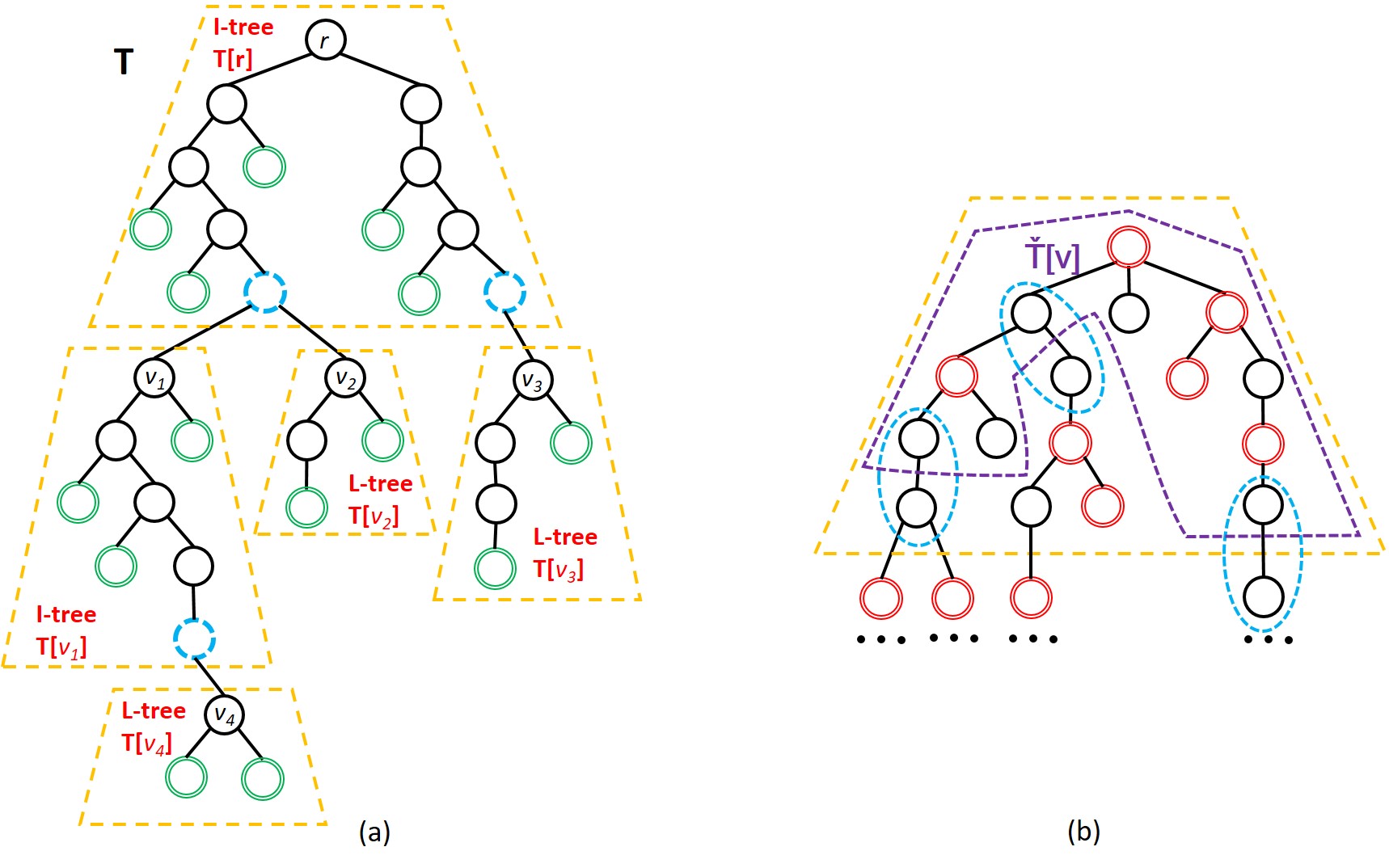}
\caption{(a) A decomposition of the tree $T$ into subtrees for $\ell=1$, $k=5$.
Layer-leaves are marked by a blue dashed circle, and
real leaves are marked by a green double circle.
{\hskip 1cm}
(b) A \cuttree, $k=5$. The client vertices of $S$ are drawn as double red
circles. The cuts along root-to-root paths are marked by blue dashed elypses.
The peak-tree $\tT[v]$ is marked by a purple dashed curve.}
\label{fig:sub+cut}
\end{figure}

The leaves of a subtree $T[v]$ can be classified into {\em real leaves} and
{\em layer-leaves}, namely, leaves of $T[v]$ that are internal nodes in $T$.
A subtree that has no other subtree below it (namely, all of whose leaves
are real) is called a {\em leaf-subtree} or simply \leaftree.
Otherwise, it is called an {\em internal-subtree} or \inttree.
(See Fig. \ref{fig:sub+cut}(a).)
We partition the vertices of $T$ into two subsets.
Let $\lleaves$ be the set of all layer-leaves,
and $\main$ be the set of all remainig vertices.
This induces a partition of the vertices of each subtree $T[v]$ into
$\lleaves[v]$
and $\main[v]$.
(For an \leaftree, $\lleaves[v]=\emptyset$.)

During the recurrent stage, each instance consists of a set $S$ of clients. This induces additional distinctions on the tree structure.
Internal subtrees are classified into two types.
The \inttree\ $T[v]$ is called a {\em \cuttree} if on every path from $v$ to a
root $w$ hanging from a layer-leaf of $T[v]$ there are two consecutive vertices
that do not belong to $S$. See Fig. \ref{fig:sub+cut}(b).
The figure also illustrates the fact that in a \cuttree\ $T[v]$
one can identify a subtree $\tT[v]$, referred to as the {\em peak} of $T[v]$,
with the property that for every edge $(x,y)$ connecting a vertex $x\in\tT[v]$
to a child $y\notin\tT[v]$, both $x,y\not\in S$.
This implies that nodes {\em below} $\tT[v]$ cannot help in dominating
clients in $\tT[v]$, namely, taking them into $D$
cannot dominate client vertices in $\tT[v]$.
$T[v]$ is a {\em \fulltree} if it is not a \cuttree, namely, there is
at least one path from $v$ to a root $w$ hanging from some layer-leaf of $T[v]$
with no two consecutive vertices of $V\setminus S$.

The idea behind the approximation scheme is as follows.
Our algorithm solves the $\CDS$ problem {\em separately},
in an optimal manner, on each subtree $T[v]$ of depth at most $k$
for the client set $S[v]$.
This can be done in time $O(k)$, but might entail inaccuracies.
As illustrated in the lower bound of Sect. \ref{ss:lb},
the main hindrance to the accuracy of a local distributed algorithm for $\CDS$
stems from long paths with a periodic occurrence of client vertices.
Such a path, threading its way along some root-to-leaf path in $T$,
might be handled poorly by the local computations.
Our goal is to bound the loss by at most 1 per subtree in the decomposition.
This is justified for \fulltrees, since in a \fulltree\ the optimum solution
$D^*$ must also use $\Omega(k)$ dominators to cover all the clients,
so the inaccuracy ratio is just $1/\Omega(k)$.

This approach is made complicated due to the fact that some subtrees
are not full, and may require only a small number of dominators.
For such subtrees (specifically, \leaftrees\ and \cuttrees),
we cannot allow the algorithm to ``waste'' more than the optimum solution.
Hence when comparing the number of dominators used by the algorithm
to that of the optimum $D^*$, we must use an accounting method
that will equate the costs over \leaftrees\ and \cuttrees, and charge
all the ``waste'' to \fulltrees.

This is done as follows. In a first phase, we locally solve the problem
optimally in each \leaftree\ and \cuttree. This is only used in order to decide,
for each such subtree $T[v]$, whether the root's parent, denoted $\parent(v)$,
must belong to the dominating set. This is important since these vertices cover
the ``interference layers'' between adjacent subtrees. 
For the \fulltrees, an optimal local solution cannot be computed.
Therefore, we artificially impose a ``waste'' in every \fulltree\ $T[v]$,
by selecting the parent of its root, $\parent(v)$, as a dominator,
whether or not necessary.
As explained above, this ``waste'' is justified by the fact that
$D^*$ must also use $\Omega(k)$ dominators in these subtrees.
As a result, when we compute a dominating set for the remaining
undominated clients in the second phase of the algorithm,
the solution computed by the algorithm on each subtree $T'$
is no greater than the number of $D^*$ dominators in $T'$.

\subsubsection*{Optimal procedure $\PROC$}
~  
A simple procedure $\PROC$ we use is an optimal algorithm for $\CDS$ on rooted trees, which runs in time $O(\mathsf{depth}(T))$ on a tree $T$.
The algorithm starts with an empty set of dominators $D$ and works its way
from the leaves up, handling each node $w$ only after it finishes handling
all its children.
It adds $w$ to the set $D$ in one of the following two cases:
\\
(1) Some of $w$'s children are clients and are not yet dominated, or
\\
(2) $w$ itself is an undominated client and is the root.

It is easy to verify that this algorithm yields a minimum cardinality solution
for $\CDS$. It is also easy to implement this greedy algorithm
as an $O({\sf depth}(T))$ time distributed protocol.

\inline Modification for subtrees:
When applying this procedure to a subtree $T[v]$ of $T$ where $v$
is not the root of $T$, we make the following small but important modification:
When the procedure reaches $v$ itself, if $v\in S$ and $v$ is still
non-dominated, then we add $\parent(v)$ instead of $v$ to the solution.
(This can be done since $\parent(v)$ belongs to the tree $T$,
although it is outside the subtree $T[v]$.)

\subsubsection*{Approximation algorithm $\AA$}
~  
\begin{algorithmic}[1]
\hrule 
\Procedure{\AA}{}
\For {every subtree $T[v]$}
    \State Decide if it is an \leaftree, a \cuttree or a \fulltree
    \State $D^{\lleaves[v]} \gets \emptyset$
    \State $\RLT \gets \{v \mid T[v]~\mbox{ is an \leaftree}\}$ { \hspace{3pt} (* \leaftree\ roots    *)}
    \State $\RCT \gets \{v \mid T[v]~\mbox{ is a \cuttree}\}$ { \hspace{3pt}  (* \cuttree\ roots    *)}
    \State $\RFT \gets \{v \mid T[v]~\mbox{ is a \fulltree}\}$ {\hspace{3pt} (* \fulltree\ roots    *)}
    \State $\RT \gets \RLT\cup\RCT\cup\RFT$ {\hspace{3pt} (* all subtree roots    *)}
\EndFor

\Statex  \ \ \ \ \ {\bf (* First dominator selection phase *) } 
\label{step:3}
\For {every \leaftree\ $T[v]$}
    \State Apply Procedure $\PROC$ to $(T[v],S[v])$
    \If{$\parent(v)\in T[w]$ was selected as a dominator}
        \State let $D^{\lleaves}[w] \gets D^{\lleaves}[w] \cup \{\parent(v)\}$
    \EndIf
\EndFor
\label{step:4}
\For{every \cuttree\ $T[v]$}
    \State Apply Procedure $\PROC$ to the peak-tree $(\tT[v],S\cap\tT[v])$
    \If{$\parent(v)\in T[w]$ was selected as a dominator}
        \State let $D^{\lleaves}[w] \gets D^{\lleaves}[w] \cup \{\parent(v)\}$
    \EndIf
\EndFor 
\label{step:5}
\For{every \fulltree\ $T[v]$, with $\parent(v)\in T[w]$ }
    \State $D^{\lleaves}[w] \gets D^{\lleaves}[w] \cup \{\parent(v)\}$
    \State $D^{\lleaves} \gets \bigcup_{v\in\RT} D^{\lleaves}[v]$
    \State let $S'$ be the set of all vertices that are dominated by $D^{\lleaves}$
    \State $S''\gets S\setminus S'$ \ {\hspace{3pt} (* Remaining undominated clients  *)}
\EndFor

\Statex \ \ \ \ \ {\bf (* Second dominator selection phase *) }
\For{every subtree $T[v]$} 
    \State Apply Procedure $\PROC$ to $(T[v],S''[v])$
    \State Let $D^{\main}[v]$ be its output set of dominators {\hspace{3pt} (* these are all internal nodes *)}
\EndFor 
\State $D^{\main} \gets \bigcup_{v\in\RT} D^{\main}[v]$
\State \Return $D^{\cA} \gets D^{\lleaves} \cup D^{\main}$
\EndProcedure
\hrule
\end{algorithmic}

\subsubsection*{Analysis}
~  
For an instance $(T,S)$ of $\CDS$, a set $D$ is said to be an {\em upmost dominating set} if it has the following property:
For every $w\in D$, replacing $w$ by $\parent(w)$ results in a non-dominating set.
(This property also suggests a natural bottom-up process for transforming
a solution $D$ into an upmost solution $D'$ of the same size.)

Denote the optimum solution by $D^*$. Without loss of generality we may assume that $D^*$ is an upmost dominating set.
The following is immediate from the definition of upmost dominating sets.

\begin{observation}
\label{obs: upmost: x in D has needy child}
Consider an instance $(T,S)$ of $\CDS$ and an upmost dominating set $D$ for it.
If $v\in D$, then there exists
some child $v'$ of $v$
in $T$ such that $v'\in S$ and $v$ is its only dominator
(or in other words, no child of $v'$ is in $D$).
\end{observation}

\begin{observation}
\label{obs:unique optimal upmost}
For any instance $(T,S)$ of $\CDS$, the dominating set selected by Procedure
$\PROC$ is equal to the unique optimum upmost solution $D^*$.
\end{observation}

We further partition the dominators of $D^*[v]$ into subsets,
according to whether they are layer-leaves or internal nodes,
and identify also the set of all {\em external} dominators,
namely, dominators that are either outside $T[v]$ or layer-leaves.
 
$D^{*,\lleaves}[v] = D^* \cap \lleaves[v]$,
\hspace{20pt}
$D^{*,\main}[v]   = D^* \cap \main[v]$,
\hspace{20pt}
$D^{*,ext}[v]   = D^* \setminus D^{*,\main}[v]$,

$D^{*,\lleaves} = \bigcup_{v\in\RT} D^{*,\lleaves}[v]$,
\hspace{20pt}
$D^{*,\main} = \bigcup_{v\in\RT} D^{*,\main}[v]$.

We also partition the vertices in each set $D^{\lleaves}[v]$ into two subsets.
Let
\begin{align*}
D_{C,L}^{\lleaves}[v] = & \{ w \mid w ~\mbox{ was added to }~ D^{\lleaves}[v]
~\mbox{ in Steps \ref{step:3} and \ref{step:4} of the algorithm} \},
\\
D_F^{\lleaves}[v] = & \{ w \mid w ~\mbox{ was added to }~ D^{\lleaves}[v]
~\mbox{ in Step \ref{step:5} of the algorithm} \},
\end{align*}

\hspace{45pt}
$D_{C,L}^{\lleaves} = \sum_{v\in RT} D_{C,L}^{\lleaves}[v]$,
\hspace{20pt}
$D_{F}^{\lleaves} = \sum_{v\in RT} D_{F}^{\lleaves}[v]$.

\begin{observation}
\label{obs:main, root, parent of L tree}
For every $v\in\RLT$, where $z=\parent(v)\in T[w]$, \\
(a) 
$D^{\lleaves}[v] = D^{*,\lleaves}[v] = \emptyset$,
and \\
\hspace{10pt}
(b) $D^{\main}[v] = D^{*,\main}[v]$.
\end{observation}

\begin{proof}
Claim (a) follows trivially since \leaftrees\ have no layer-leaves,
so $\lleaves[v]=\emptyset$.

Claim (b) follows from the observation that for an \leaftree\ $T[v]$,
both $D$ and $D^*$
induce optimum upmost dominating sets for $T[v]$, and these induced
dominating sets, $\tilde D$ and ${\tilde D}^*$,
are identical by Obs. \ref{obs:unique optimal upmost}.
\end{proof}

It may be instrumental to pause and make the following observation
concerning the sets $\tilde D$ and ${\tilde D}^*$ discussed in the above proof.
For the purpose of dominating the clients of $S[v]$, either both sets contain
$z=\parent(v)$ or both do not. One might hope that this will establish that
$D^{*,\lleaves}[w] = D^{\lleaves}[w]$.
However, this argument is false,
since we need to account for the possibility that one of the dominating sets
($D$ or $D^*$) includes $z$ in order to dominate another client child $v'$,
other than $v$, while the other dominates $v'$ in some other way, 
and does not include $z$.
Nevertheless, we can prove the following weaker properties, which suffice for our purpose.

\begin{lemma}
\label{cor:D vs D* for lleaves}
$D_{C,L}^{\lleaves}[v] ~\subseteq~ D^{*,\lleaves}[v] ~\subseteq~ D^{\lleaves}[v]$~~~
for every $v\in\RT$.
\end{lemma}

\def\PROOFTA{
To prove the second containment, suppose $z \in D^{*,\lleaves}[v]$.
As $D^*$ is an upmost dominating set for $T$,
by Obs. \ref{obs: upmost: x in D has needy child},
$z$ must have some child $v'$ such that $v'\in S$ and $v'$ is not dominated
by any of its children it $T$.
We argue that this $v'$ forces $z$ to be in $D^{\lleaves}[v]$.
To see this, consider the following three cases.
\begin{itemize}
\item
If $v'\in\RLT$, then both $D\cap(T[v']\cup\{z\})$ and $D^*\cap(T[v']\cup\{z\})$
are optimum upmost dominating sets for $T[v']$, hence they are identical
by Obs. \ref{obs:unique optimal upmost}, and therefore $z\in D$, implying
that $z \in D^{\lleaves}[v]$.
\item
If $v'\in\RCT$, then both $D\cap(\tT[v']\cup\{z\})$ and
$D^*\cap(\tT[v']\cup\{z\})$ are optimum upmost dominating sets for $\tT[v']$,
and $z \in D^{\lleaves}[v]$ by the same argument.
\item
If $v'\in\RFT$, then $z=\parent(v') \in D^{\lleaves}[v]$ by Step \ref{step:5}
of the algorithm.
\end{itemize}

To prove the first containment, suppose $z \in D_{C,L}^{\lleaves}[v]$.
Then $z$ was added in Step \ref{step:3} or \ref{step:4} of the algorithm
to the set $D^{\lleaves}[v]$ since it belonged to the dominating set $D[v']$
generated by Procedure $\PROC$ for the subtree $T[v']$ for some vertex
$v'\in\RLT\cup\RCT$.
As the procedure always generates an upmost dominating set,
it follows that $v'\in S$, and after selecting $D[v']$,
$v'$ is still not dominated.
As $D^*$ also induces an upmost dominating set for $T[v']$,
the same holds for $D^*[v']$, hence $\parent(v')=z$ must be in $D^*$ as well,
i.e., $z \in D^{*,\lleaves}[v]$.
} 

\begin{observation}
\label{obs:main of C tree}
For every $v\in\RCT$, ~~ $D\cap \tT[v] = D^*\cap \tT[v]$.
\end{observation}

\begin{proof}
The claim follows from the observation that for a \cuttree\ $T[v]$,
both $D$ and $D^*$
induce optimum upmost dominating sets for $\tT[v]$, and these induced
dominating sets, $\tilde D$ and ${\tilde D}^*$,
are identical by Obs. \ref{obs:unique optimal upmost}.
\end{proof}

We make use of the following straightforward monotonicity property.

\begin{observation}
\label{obs:monotonicity}
For every rooted tree $T$ and two client sets $S_1\subseteq S_2$,
the corresponding optimum dominating sets $D_1$ and $D_2$,
for $(T,S_1)$ and $(T,S_2)$ respectively, satisfy
$|D_1| \le |D_2|$.
\end{observation}

\begin{lemma}
\label{lem:main}
$|D^{\main}[v]| \le |D^{*,\main}[v]|$
for every $v\in\RFT\cup\RCT$.
\end{lemma}

\def\PROOFTB{
Recall that $S'[v]$ is the set of clients from $S[v]$ that were dominated by
the set of dominators $D^{\lleaves}$ selected in the first phase.
Let $S^{*'}[v]$ be the set of clients from $S[v]$ that are dominated by
$D^{*,ext}[v]$.

\begin{claim}
\label{cl:contained}
$S^{*'}[v] ~\subseteq~ S^{'}[v]$.
\end{claim}
\begin{proof}
Consider some client $z\in S^{*'}[v]$, which is dominated by some external
dominator in $D^{*,ext}[v]$. We need to show that $z$ is also dominated
by some external dominator in $D^{\lleaves}$, so $z\in S^{'}[v]$.

\begin{figure}[htb]
\centering
\includegraphics[width=2.3in]{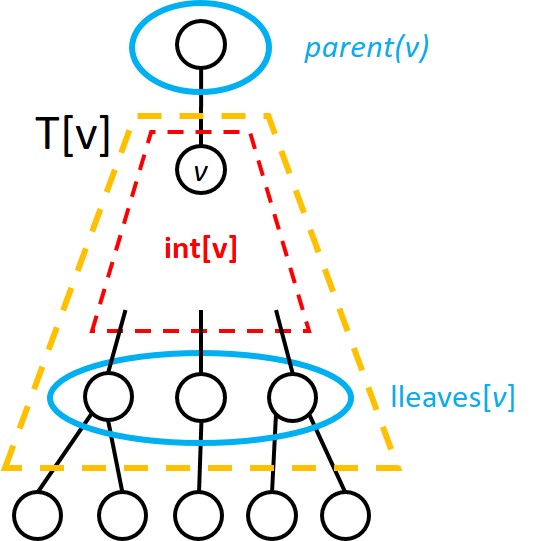}
\caption{The potential external dominators of $\main[v]$.}
\label{fig: ext-dom}
\end{figure}

\noindent
Note that the only internal nodes of $\main[v]$ that can potentially
be dominated by external nodes dominators are the root $v$,
which can be dominated by $\parent(v)$,
and the parents of nodes in $\lleaves[v]$, which can be dominated by
their children. 
Hence there are two cases to consider.

(1) $z=v$:
Then $z$ is dominated in $D^{*,ext}[v]$ by $\parent(z)$, which is its only
potential external dominator.
This implies that $\parent(z) \in D^{*,\lleaves}[w]$ for some $w$.
By the second containment in Lemma \ref{cor:D vs D* for lleaves},
also $\parent(z)\in D^{\lleaves}[w]$, so $z\in S^{'}[v]$.

(2) $z=\parent(w)$ for some $w\in D^{*,\lleaves}[v]$:
Then $z$ is dominated in $D^{*,ext}[v]$ by $w$.
By the second containment in Lemma \ref{cor:D vs D* for lleaves},
also $w\in D^{\lleaves}[v]$,
so $D^{\lleaves}$ dominates $z$ as well, hence $z\in S^{'}[v]$.
\end{proof}

Recall that $S''[v]=S[v] \setminus S'[v]$, the clients from $S[v]$
that were not dominated by the end of the first phase.
Letting $S^{*''}[v]=S[v] \setminus S^{*'}[v]$,
we have by Claim \ref{cl:contained} that
$$S^{''}[v] ~\subseteq~ S^{*''}[v].$$
The clients of $S^{*''}[v]$ must be dominated by $D^{*,\main}[v]$,
and the clients of $S^{''}[v]$ are dominated by $D^{\main}[v]$.
The lemma now follows from Obs. \ref{obs:monotonicity}.
} 

Lemma \ref{lem:main} and Obs. \ref{obs:main, root, parent of L tree}(b)
imply 
$|D^{\main}| \le |D^{*,\main}|$.
Combining with the first containment in Lemma \ref{cor:D vs D* for lleaves}, 
$$|D^{\cA}| - |D_{F}^{\lleaves}| ~=~ |D^{\main}| + |D_{C,L}^{\lleaves}|
~\le~ |D^{*,\main}| + |D^{*,\lleaves}| = |D^*|.$$
Denote by $t_f$ (respectively, $t_c$) the number of \fulltrees\ 
(resp., \cuttrees) in the decomposition of $T$.
Noting that $|D_F^{\lleaves}| = t_f$, we get that
$|D^{\cA}|
\le |D^*|+t_f$.
Since a full tree contains a path on which, for two consecutive vertices, at least one is in $S$, 
it is immediate that
$|D^*[v]| \ge (k-1)/4$ for every \fulltree\ $T[v],$
and therefore
$|D^*| \ge t_f \cdot (k-1)/4$.
It follows that the approximation ratio of the algorithm satisfies
$$\rho ~\le~ \frac{|D^{\cA}|}{|D^*|} ~\le~ \frac{|D^*|+t_f}{|D^*|}
~=~ 1 + \frac{t_f}{|D^*|} ~\le~ 1 + \frac{t_f}{t_f \cdot (k-1)/4}
~=~ 1 + \frac{4}{k-1}~.$$

We get the following result.

\begin{theorem}\label{ctas:tree:thm}
For every positive integer $k$, there exists a deterministic distributed local approximation algorithm for $\CDS$, with preprocessing allowed, whose online runtime on every $n$-vertex tree and every instance is at most $O(k)$ with approximation ratio of at most $1 + \frac{4}{k-1}$.
\end{theorem}

\subsection{A CTAS for \CDS ~on Planar Graphs}
\label{sec:ctas-planar}
\subsubsection{Constant Approximation for \CDS~ on Planar Graphs}

The state of the art algorithm for constant round planar dominating set approximation in the $\LOCAL$ model achieves an approximation ratio of $20$ by a recent work of Heydt et al. \cite{heydt2021local}. Their algorithm and analysis extend to the client dominating set problem with slight modifications. See Algorithm \ref{alg:dom:planar:const} for the pseudocode.

\begin{algorithm}
\caption{Constant Approximation for MCDS in Planar Graphs}
\label{alg:dom:planar:const}
    \begin{algorithmic}[1]
        \State $C \gets$ client set
        \State For every $A \subseteq V(G)$, define $N_C(A) = \{w \mid w \in C \text{ and } (w, v) \in E(G) \text{ for some } v \in A \}$
        \\
        ~\hbox{\hskip 4.2cm}
        $N_C[A] = N_C(A) \cup (A \cap C)$
        \State $D_1 \gets \{v \in V(G) \mid \forall A \subseteq V(G) \setminus \{v\}, N_C[A] \supseteq N_C(v) \Rightarrow |A| \geq 4\}$
        \State For every $v \in V(G)$, compute $B_v ~=~ \left\{w \in V(G) \setminus D_1 ~\bigm|~  | N_C(v) \cup N_C(w) | \geq 10 \right\}$
        \State $D_2 \gets \left\{v \in V(G) \setminus D_1 \bigm| B_v \neq \emptyset \right\}$
        \State $D_3 \gets C \setminus N_C[D_1 \cup D_2]$
        \State Return $D_1 \cup D_2 \cup D_3$
    \end{algorithmic}
\end{algorithm}

\begin{theorem}
\label{thm:39-approx}
    Algorithm \ref{alg:dom:planar:const} provides a $39$-approximation for the MCDS problem in planar graphs.
\end{theorem}

\begin{proof}
    The proof outline is almost same as that in \cite{heydt2021local}. Let $D^*= \{b_1, b_2, \dots b_{|D^*|}\}$ be some optimal solution for a given MCDS instance. Define the set 
    \begin{equation}
    \label{eq:def hatD}
    \hat{D} ~~=~~ \{v \in V(G) ~~\bigm|~~ \mbox{ for every }~ A \subseteq D^* \setminus \{v\},~~ N_C[A] \supseteq N_C(v) \Rightarrow |A| \geq 4 \}.
    \end{equation}
    
    Observe that $\hat{D}$ is defined similarly to the set $D_1$ constructed in the algorithm, except that $V(G)$ is replaced with $D^*$. Every element in $D_1$ must also belong to $\hat{D}$, so 
    \begin{equation}
    \label{eq:D1 Dhat}
    D_1 \subseteq \hat{D}    
    \end{equation}
    
    \begin{lemma}
        $|\hat{D} \setminus D^*| < 4 |D^*|$
    \end{lemma}
    
    \begin{proof}

        Suppose, for the sake of contradiction, that $|\hat{D} \setminus D^*| \geq 4 |D^*|$. Then there exists an independent set of size at least $|D^*|$ in the graph induced by $\hat{D} \setminus D^*$, as every subgraph of a planar graph is $4$-colorable.
        Let $I = \{a_1, a_2, \dots a_{|D^*|}\}$ be an arbitrary such independent set.
        For every client $c \in C$, let $f(c)$ be the smallest integer such that $b_{f(c)}$ dominates $c$.
        
        Let $G'$ be the graph obtained by contracting all edges $(c,~b_{f(c)})$ in $G$, for every $c \in C \setminus (I \cup D^*)$. The underlying simple graph induced by $I \cup D^*$ in the graph $G'$ is bipartite, and every vertex in $I$ has degree at least $4$. 
        Denoting the number of vertices and edges in this bipartite graph by $n$ and $m$, respectively,
        we get $m \geq 4 \cdot \frac{n}{2} \geq 2n$. However, every simple planar bipartite graph satisfies $m < 2n$, yielding the desired contradiction.
    \end{proof}
    
    \begin{lemma}
    \label{lem: 3D*}
        $|D_2 \setminus (D^* \cup \hat{D})| \leq 3 |D^*|$
    \end{lemma}
    
    \begin{proof}
        Consider any vertex $v \in D_2$ such that $v \not \in D^*$. By definition (See Line 6 in Algorithm \ref{alg:dom:planar:const}), $v \not \in D_1$, so by the definition of $D_1$  there must exist a set of size at most $3$ that dominates the client neighbors of $v$. Let $A_v = \{y_1, y_2, y_3\}$  be any such set for some $y_1, y_2, y_3$  (that need not be distinct). 
        
        \begin{claim}\label{clm:bvcontainment}
        For every $v$, the set $B_v$ computed by the algorithm satisfies $B_v \subseteq A_v$
        \end{claim}
        \begin{proof}
            Suppose, for the sake of contradiction, that there exists some $w \not \in A_v$ belonging to $B_v$. By the definition of $B_v$, $w$ and $v$ share (at least) $10$ common clients, $C''$. Note that $C''$ does not include $w$ and $v$. Moreover, $C''$ must also be dominated by the vertices of $A_v$, hence at least one of the vertices in $A_v$ must dominate at least $\lceil 10/3 \rceil = 4$ of these $10$ clients. Suppose this vertex is $y_1$. By the above discussions, we must have $|N_C(v) \cap N_C(w) \cap N_C(y_1)| \geq 3$, which implies the existence of $K_{3, 3}$ as a subgraph, contradicting the planarity of the graph. 
        \end{proof}
        
        The $v-w$ relation $v \in B_w$ is symmetric, so we can split $D_2$ as,
        \begin{eqnarray*}
        D^1_2 &=& \bigcup_{v \in D^* \setminus D_1} \{v\}\cup B_v~,
        \\
        D^2_2 &=& \bigcup_{v \in \hat{D} \setminus (D^* \cup D_1)} \{v\} \cup B_v~, \mbox{ and}
        \\
        D^3_2 &=& \bigcup_{v \not \in (\hat{D} \cup D^* \cup D_1)} \{v\} \cup B_v~.
        \end{eqnarray*}
        
        \begin{claim}
            $D^3_2 \subseteq D^1_2$
        \end{claim}
        \begin{proof}
            Consider a vertex $v' \in D^3_2$. Then
            there exists some vertex $v$ such that $v \not \in (\hat{D} \cup D^* \cup D_1)$ and $v'\in \{v\} \cup B(v)$. 
            Since $v\not\in \hat{D}$, by Eq. \eqref{eq:def hatD} there exists a set $A_v = \{b_1, b_2, b_3\} \subseteq D^*$ that dominates $N_C(v)$. By symmetry, if $b_i \in B_v$ then $v \in B_{b_i}$ and therefore $v$ and $B_v$ are included in $D^1_2$, so $v' \in D^1_2$. 
        \end{proof}
        
        \begin{claim}
            $D^2_2 \setminus \hat{D} = \emptyset$
        \end{claim}
         
        \begin{proof}
            Suppose, for sake of contradiction, that there exists some $w \in D^2_2 \setminus \hat{D}$. There must exist $v \in \hat{D}\setminus(D^* \cup D_1)$ such that $w \in B_v$. By symmetry $v \in B_w$. As $w \not\in \hat{D}$, there exists a set $A_w \subseteq D^*$ with $|A_w| \leq 3$ that dominates $N_C(w)$. From Claim \ref{clm:bvcontainment}, $B_w \subseteq A_w \subseteq D^*$. This implies that $v \in D^*$ which is a contradiction.
        \end{proof}
        
        Finally we have $D_2 \setminus (D^* \cup \hat{D}) \subseteq \cup_{v \in D^* \setminus D_1} B_v$ and since $|B_v| \leq |A_v| \leq 3$, we have 
        $|D_2 \setminus (D^* \cup \hat{D})| \leq 3 |D^*|$, 
        completing the proof of Lemma \ref{lem: 3D*}.
        \end{proof}
    \begin{lemma}
        If $v \not \in D_1 \cup D_2$, then $|N_C(v)| \leq 30$
    \end{lemma}
    
    \begin{proof}
        Suppose, for the sake of contradiction, that there is some vertex $v\not\in D_1 \cup D_2$ such that $N_C(v) \geq 31$. By the definition of $D_1$, as $v \not \in D_1$, there exists a set $A \subseteq V(G) \setminus \{v\}$ of size at most $3$, that dominates all clients of $v$, and therefore at least one vertex $w\in A$ dominates at least $\lceil 31 / 3 \rceil = 11$ clients. We must have $|N_C(v) \cup N_C(w)| \geq 10$ and therefore $v \in D_2$, leading to contradiction.
    \end{proof}
    
    The above lemma shows that after removing clients that are dominated by $D_1 \cup D_2$, every other vertex can dominate at most $30$ clients. Therefore, the set $D_3$ constructed in the last step of the algorithm, which takes all the remaining undominated clients to the dominating set, must be at most $31$ times the optimal, i.e., $|D_3 \setminus D^*| \leq 31$. Putting the lemmas together, we can bound size of $D = D_1 \cup D_2 \cup D_3$ as $|D| \leq |D^*| + |\hat{D}\setminus D^*| + |D \setminus (\hat{D} \cup D^*)| \leq |D^*| + |\hat{D}| + |D_2 \setminus(D^* \cup \hat{D})| + |D_3 \setminus D^*| \leq 39 |D^*|$, proving Theorem \ref{thm:39-approx}.
\end{proof}

\def\wt{{\sf wt}}

\subsubsection{\texorpdfstring{A $1 + \epsilon$}{1+epsilon} approximation}

We adapt the distributed $1 + \epsilon$-approximation scheme of Czygrinow et al.~\cite{CHW08}, whose round complexity is $O(\left(\frac{1}{\epsilon}\right)^c\log ^* n)$ where $c = \log_{24/23} 3$. We first provide a high level overview of their algorithm and the major differences and difficulties towards adapting it to the recurrent $\CDS$ problem with preprocessing.

\textbf{Fast $\LOCAL$ Algorithm}. The graph $G$ is partitioned into several disjoint connected components (called clusters) such that (i) each cluster has diameter at most $d$, and (ii) the total number of edges crossing two clusters is at most $E$. All the cross edges are then removed and the dominating set problem is solved optimally and independently within each cluster. If the cluster diameter $d$ is small enough, then the previous step requires only $O(d)$ rounds, as the entire graph can be collected at some delegated leader who can then solve the problem locally. If the number of cross edges $E$ is small enough, then we get a good approximation of the dominating set.

For finding a good clustering, the algorithm first makes use of a constant approximation which can be obtained in constant rounds. Clustering around the dominators computed by the constant approximation results in clusters with diameter $d \leq 2$. Each cluster is then contracted into a single node. Let $G_0$ be the obtained underlying simple graph. Observe that $G_0$ has at most $39 |D^*|$ vertices, where $D^{*}$ is the optimal client dominating set. The graph $G_0$ is initially weighted with each edge having weight $1$.

Now suppose we are able to cluster $G_0$ into connected components $G'_1, G'_2, \dots G'_s$ so that the total weight of edges crossing clusters, $|E_{\mathsf{cross}}|$, is at most $\epsilon$ of its initial total (which is at most $3 \cdot 39|D^*|$ by planarity). Let $D$ be the union of set of dominators obtained by solving each graph $G'_i$ independently and optimally. We argue that $D$ is $1 + \epsilon'$ approximation. Let $D_2$ be the constant approximation obtained in the first step. Consider the set $\hat{D} = D^{*} \cup \{v \mid (u, w) \in E_{\mathsf{cross}} \text{ and } v \text{ dominates } w\}$. $\hat{D}$ is a valid dominating set and moreover, $\hat{D} \cap V(G'_i)$ dominates all clients in $V[G'_i]$ for every $i$.
Since $D$ was obtained by solving $G'_i$ optimally, we have $|D \cap V(G'_i)| \leq |\hat{D} \cap V(G'_i)$. Adding up over all clusters, we get $|D| \leq \hat{D}$.
However $|\hat{D}| \leq |D^*| + 2 \cdot (\epsilon \cdot (3 \cdot 39 |D^*|))$.
Plugging in we get $|D| \leq (1 + 234 \epsilon) |D^*|$.

A clustering of $G_0$ is computed by repeatedly applying a contraction process. The contraction process for a weighted graph $G$ is as follows. A large weight subset of the edges of $G$ is chosen and then oriented such that every node has out-degree at most $1$. Such oriented graphs are called \emph{pseudo-forests}. For a planar graph, it is possible to choose in one round a pseudo-forest that has at least $\frac{1}{6}^{th}$ the total weight of all edges. The pseudo forest is then $3$-colored using the Cole-Vishkin Algorithm. The $3$-coloring is used to split the forest into disjoint stars (graphs with diameter at most $2$), while not losing more than a quarter (in weight) of the edges of the pseudo-forest. Each star is then contracted into a single node. After contraction, it is possible that the graph has multiple edges. All multiple edges between a pair of nodes are replaced by a single edge whose weight is set to equal their total weight.
The above contraction process is applied repeatedly until the weight of the edges reduces to $\epsilon$ of the initial total. Since each contraction removes at least $\frac{1}{24}$ of the edges, it is sufficient to repeat the process $t = O(\log_{24/23}{\frac{1}{\epsilon}})$ times. 

Let the final graph obtained be $G_t$. $G_t$ provides a clustering of the original graph $G$, which can be obtained by uncontracting all the edges. The number of cross edges of the clustering is the weight of $G_t$. Each time a star is contracted, the diameter of the corresponding clusters increases by a multiplicative factor of at most $3$ and so the diameter 
of each cluster given by $G_t$ 
is $O(3^t)$.

The most time consuming step in this process is that of $3$-coloring the pseudo-forest, which takes $O(3^{i} \log^*n)$ rounds during the $i^{th}$ iteration of the contraction process. The other operations take $O(3^i)$ rounds. The total round complexity is $O(\sum\limits_{i=0}^t 3^i \log^* n) = O(3^t \log^* n)$.

\textbf{Adapting to $\CDS$.} First, we remove the edges that are not incident on any client. These edges do not contribute to the criteria for a set to be a dominating set, and they can be ignored. We then compute a constant approximation $\tilde{D}$ as per Algorithm \ref{alg:dom:planar:const}. The initial clustering is obtained by choosing for each client $c$ an arbitrary dominator of $c$ from $\tilde{D}$ and contracting the edge between them. Additionally, every vertex that is neither a client nor a dominator chooses an arbitrary neighboring client and the edge between them is merged. The remaining steps are identical to the previous procedure.

\textbf{Speeding up using a preprocessing phase.} One potential preprocessing operation that may improve the round complexity of the online stage might be to compute a proper $4$-coloring of the planar graph. Unfortunately, while a coloring of any graph remains valid after the removal of edges or vertices, it does not remain valid after contractions. An arbitrary precomputed coloring might not be of much use in coloring the contracted graphs that arise from repeated contractions. To accommodate contractions, we precompute a \textit{non-repetitive} coloring of $G$ (which is the only output of our preprocessing phase). A \textit{non-repetitive} coloring is a coloring of the graph such that for any even length simple path, the ordered set of colors in the first half of the path is different from that of the second half. Non-repetitive colorings were first proposed by Alon et al \cite{alon2002nonrepetitive}. The minimum number of colors required to realise a non-repetitive coloring is called the Thue number of the graph and is denoted by $\pi(G)$. Dujmovi\'c et al. \cite{Dujmovi__2020} showed recently that $\pi(G) \leq 768$ for all planar graphs $G$.

Suppose we have a pseudo forest $F$ that needs to be $3$-colored and suppose $F$ is obtained from $G_t$, i.e., after $t$ iterations of the contraction process. Let $\out(v)$ denote the other end of the outgoing edge of $v$ in $F$. In order to $3$-color the forest, it is sufficient to choose colors in such a way that $\out(v)$ and $v$ have different colors, for every $v$. We can associate with each node $v$ of $G_t$, a connected component (denoted $G_v$) in the original graph $G$ that contains the ends of all edges that were contracted to $v$. Choose any edge $e$ that crosses $G_v$ and $G_{\out(v)}$. Construct a spanning tree of $G_v$ and root it at the endpoint $r(v)$ of $e$ that lies in $G_v$. We now color $v$ with the ordered set of non-repetitive colors traced on the unique path from $r(v)$ to $r(\out(v))$, excluding $r(\out(v))$, in the graph $G_v \cup G_{\out(v)} \cup \{e\}$. We enumerate these colors from $1$ to $728^{d+1}$ where $d$ is the maximum diameter of the clusters. Let the computed path be $P_v$.
Observe that whenever $\out(\out(v)) \neq v$, the paths $P_v$ and $P_{\out(v)}$ can be concatenated to form a simple path in the graph $G$. If $P_v$ and $P_{\out(v)}$ have different lengths, then the colors assigned to them are different. Otherwise, by the property of a non-repetitive coloring, the ordered set of colors of $P_v$ and $P_{\out(v)}$ must be different. When $\out(\out(v)) = v$, we have a 2-cycle. In this case we color one of the nodes $\{v, \out(v)\}$ (whichever has higher id, say $v$) with its own non-repetitive color and redefine $P_{v} = \{\out(v)\}$. Now the paths $P_v$ and $P_{\out(v)}$ may be concatenated to obtain a simple path $P$. See Algorithm \ref{alg:3color:planar} for the pseudo-code.
We now have a $768^{d+1}$ coloring of the pseudo-forest $F$, which can then be reduced to a $3$-coloring using the Cole-Vishkin Algorithm. The complexity is $O(d \log^* {768^{d+1}}) = O(d \log^* d)$. This leads us to our main lemma:
\begin{lemma}
    Given a clustering of the graph $G$, Algorithm \ref{alg:3color:planar} provides a $3$-coloring of the graph obtained by contracting each cluster into a single vertex. Moreover this algorithm can be implemented as an $O(d \log^* d)$ round $\LOCAL$ protocol, where $d$ is the maximum diameter amongst the induced components of the clustering.
\end{lemma}

Algorithm \ref{alg:3color:planar} is the main unique ingredient to our adaptation of Czygrinow et al's algorithm. Plugging this component into their algorithm directly leads to an $O_{\epsilon}(1)$ $\LOCAL$ algorithm. For concreteness, the complete clustering procedure is described in Algorithm \ref{alg:cds:planar} with some minor changes to account for the clients. Once clustering is done, we proceed in the same way, i.e., solve the $\CDS$ problem optimally and independently within each cluster. Solving $\CDS$ exactly requires NP-Hard problems to be solved in the online phase, which may be undesirable. This can be fixed by replacing the optimal solution with a $PTAS$ in planar graphs for the $\CDS$ problem by a similar adaptation of Baker's algorithm \cite{baker1994approximation}.

\begin{algorithm}
    \caption{$3$-coloring pseudo-forest}
    \label{alg:3color:planar}
    \begin{algorithmic}[1]
    \Procedure{$3$-color}{}
        \Statex \textbf{Input}: 
        \Statex (i) $\col : V(G) \rightarrow [768],$ A non-repetitve coloring of the given planar graph $G$
        \Statex (ii) $\cluster : V(G) \rightarrow \mathbb{N}$, describes a partitioning of the vertices of $G$ that induce connected components of diameter at most $d$
        \Statex (iii) $G_t:$ the graph where every cluster is contracted to a single node.
        \Statex (iv) $\out : V(G_t) \rightarrow V(G_t)$, describes a pseudoforest in the graph $G_t$ \Comment{ $\out(v)$ is the other end of the unique outgoing edge from $v$}
        
        \Statex \textbf{Output}: $\col_f : V(G_t) \rightarrow [3]$, a proper $3$-coloring of the given pseudoforest
        
        \ForAll {clusters $v \in V(G_t)$ (in parallel) }
            \State $p \gets \out(v)$, the parent of $v$ in pseudo-forest of $G_t$
            \State Let $G_v, G_p$ be the connected components of $G$ that are contracted to $v, p$ in $G_t$
            \State $e_v \gets $ any edge in $G$ that crosses $G_v, G_p$ and $r_v \gets $ the end of $e$ in $G_v$
            \State $T_v \gets $ Any spanning tree of $G_v$, rooted at $r_v$
        \EndFor 
        \ForAll {clusters $v \in V(G_t)$ (in parallel)}
            \If {$\out(p) \neq v$ or $v < p$} \Comment{detect cycles of length $2$}
                \State $\path(v) \gets$ The unique path from $r_v$ to $r_p$ in the graph $T_v \cup T_p \cup \{e\}$
            \Else \Comment{Treating the case of cycle of length 2 separately}
                \State $\path(v) \gets \{r_v\}$
            \EndIf
        \EndFor
        \State $\col_f(v) \gets $ the ordered set of colors in $\path(v)$
        \State Enumerate $\col_f(v)$ using integers from $1$ to $768^{d + 1}$
        \State Reduce $\col_f(c)$ to a $3$-coloring using the Cole-Vishkin Algorithm
        \State \Return $\col_f$
    \EndProcedure
    \end{algorithmic}
\end{algorithm}

\begin{theorem}
    For every planar graph $G$,
    \begin{itemize}
        \item $\SUPTIME(\CDS_{\epsilon}, G)$ is $O(\left(\frac{1}{\epsilon}\right)^{c} \log^*{\left(\frac{1}{\epsilon}\right)})$, where $c = \log_{24/23} 3$. 
        \item Realizing the above round complexity requires only $O(1)$(i.e. a constant independent of both $\epsilon$ and $G$) additional bits to be stored in each node of ~$G$.
    \end{itemize}
\end{theorem}

\begin{algorithm}
\caption{Clustering for Planar CDS}
    \label{alg:cds:planar}
\begin{algorithmic}[1]
    \Statex \textbf{Input}: Client set $C$, a non-repetitive coloring of $G$ and $\epsilon$.
    \Statex \textbf{Output}: A $1+\epsilon$ approximation of the optimal set dominating $C$.
    \Statex \textbf{Phase 1}: Finding a good initial clustering.
        \State Remove all edges that do not have a client incident on them.
        \State Remove isolated vertices after previous step.
        \State Compute a constant approximation $D^{\star}$ for $C$ using Algorithm \ref{alg:dom:planar:const}.
        \ForAll{nodes $v \in V(G) \setminus D^{\star}$}
            \If{$v$ has a neighbor in $D^{\star}$}
                \State $u \gets $ any neighbor in $D^{\star}$
            \Else
                \State $u \gets $ any neighbor in $C$ or $\perp$ (if such a node doesn't exist)
            \EndIf
            \State Contract the edge $e = (u, v)$, if $u$ exists
        \EndFor
        \Statex \Comment{Done in parallel and implicitly, i.e., contracted vertices know their neighbors}
        \Statex \textbf{Phase 2}: Improving the clustering
        \State $G_0 \gets $ underlying simple graph obtained at end of Phase 1.
            \State Set $\mathsf{wt}(e) \gets 1$ for all $e \in E(G_0)$
        \ForAll{$t = 0, 1, \dots \lceil \log_{24/23} \frac{234}{\epsilon} \rceil$}
            \State $\mathsf{out}(u) \gets $ any neighbor $v$ such that $\mathsf{wt}((u, v))$ is maximized
            \State $H \gets $ induced by the edges $\{(\mathsf{out}(u), u) \mid u \in G_t\}$ \Comment{Heavy pseudo-forest}
            \State $\mathsf{col} \gets $ $3$-coloring of $H$ obtained using Algorithm \ref{alg:3color:planar}.
            \ForAll{$u \in H$ with $\mathsf{col}(u) = 1$ (in parallel)}
                \State $I_u, O_u \gets \{(u, v) \mid u = \mathsf{out}(v) \}, \{(u, v) \mid v = \mathsf{out}(u)\}$
                \State Remove either $I_u$ or $O_u$ from $H$, whichever has smaller total weight
            \EndFor
            \ForAll{$u \in H$ with $\mathsf{col}(u) = 2$ (in parallel)}
                \State $I_u, O_u \gets \{(u, v) \mid u = \mathsf{out}(v), \mathsf{col}(v) = 3 \}, \{(u, v) \mid v = \mathsf{out}(u), \mathsf{col}(v) = 3\}$
                \State Remove either $I_u$ or $O_u$ from $H$, whichever has smaller total weight
            \EndFor
            \State \Comment{$H$ now consists of connected components with diameter at most $10$.} 
            \State $F \gets $ rooted spanning forest of $H$
            \State $E_F, O_F \gets $ edges of $F$ at even and odd depths respectively
            \State Remove either $E_F$ or $O_F$, whichever has smaller total weight
            \State For all edges $e \in E(H)$, contract $e$ in $G_t$
            \State $G_{t+1} \gets $ underlying simple graph obtained after contractions.
            \State For all edges $e = (u, v) \in G_{t+1}$, set $\mathsf{wt}(e) \gets $ number of edges between $u, v$ after all contractions of edges in $H$.
        \EndFor
        \State \Return $G_T$
\end{algorithmic}
\end{algorithm}

\begin{proof}
    As mentioned previously, we adapt the scheme of Czygrinow et al. The high level idea is to carefully cluster the graph into components with small diameter and essentially solve the $\CDS$ problem independently within each cluster (i.e., ignoring or removing the cross edges) by a brute-force manner.

    The clustering procedure is outlined in Algorithm \ref{alg:cds:planar}. We go through the procedure and analyze it below.

    \textbf{Phase 1}: The first observation to be made is that edges with no incident client on them can be ignored. The existence or absence of these edges does not affect the correctness of any candidate solution to the $\CDS$ instance. After this removal we may get several disconnected components, which we can solve separately.

    In the initial clustering (Lines 1-11 of Algorithm \ref{alg:cds:planar}), each cluster has diameter at most $4$. This is easy to see as there is a path of length at most $2$ to some vertex in $D^{\star}$. Note that Clients are directly dominated by some vertex in $D^{\star}$ and non-clients either have a neighboring client adjacent to them or are isolated. Each vertex in $D^{\star}$ is present in its own unique cluster.

    \textbf{Phase 2}: The objective of this phase is to improve the clustering in Phase 1. Let $G_0$ be the contracted graph obtained at the end of Phase 1. By planarity we have, $|E(G_0)| \leq 3 |V(G_0)| \leq 3 \cdot 39 |D_{\mathsf{opt}}|$.
    By definition, we have $\mathsf{wt}(G_0) = |E(G_0)| \leq 117 |D_{\mathsf{opt}}|$. Here we use $\wt(G)$ to denote the total weight of all edges in $G$.

    We now describe the clustering procedure of Phase 2 (Lines 15-32). In Line 16, obtains a heavy-weight pseudo-forest subgraph of $G_0$ by a simple local greedy procedure, choose an arbitrary incident edge with maximum weight (Line 15). 
    \def\wt{\mathsf{wt}}
    \begin{lemma}
        $\mathsf{wt}(H) \geq \frac{1}{6} \mathsf{wt}(G)$
    \end{lemma}
    \begin{proof}
        We make use of Nash-Williams Theorem, i.e. since $G_t$ is planar it can be decomposed into forests $F_1, F_2, F_3$. In each of these forests, there exists an orientation such that every node has out degree at most $1$. Let the outgoing edge of $u$ in the three forests be $\mathsf{out}_1(u), \mathsf{out}_2(u), \mathsf{out}_3(u)$ and let $\mathsf{out}(u)$ be the chosen outgoing edge in Line 15. WLOG, let $F_1$ be the forest with highest weight amongst the three.  By pigeon hole principle we have, $\wt(F_1) = \sum_u \wt((\mathsf{out}(u), u)) \geq \frac{1}{3} \wt(G)$.

        The chosen edges in Line 15-16 is done for each node independently. While for the forests $F_1, F_2, F_3$, $\mathsf{out}_i(u)$ corresponded to a unique edge, this is not necessarily the case for the pseudo-forest $H$ chosen in Line 16. In particular we could have $\mathsf{out}(u) = v$ and $\mathsf{out}(v) = u$ and therefore it is not the case that $\wt(H) = \sum_u \wt((\mathsf{out}(u), u))$. However each edge $(\mathsf{out}(u), u)$ is counted at most twice in the summation from which we get, 
        \begin{equation*} \begin{split}
            \wt(H) &\geq \frac{1}{2} \sum_u \wt((\mathsf{out}(u), u)) \\
            &\geq \frac{1}{2} \sum_u \wt((\mathsf{out}_1(u), u)) \ \ \ \ [\text{By greedy choice of } \mathsf{out}(u)] \\
            & \geq \frac{1}{2} \wt(F_1) \geq \frac{1}{6} \wt(G_t)
        \end{split} \end{equation*}
    \end{proof}
    
    We next address Lines 18-25. This part of the algorithm breaks down the forest $H$ into small diameter components. This is done in two steps. In the first step, for each node with color $1$, either all its incoming or the unique outgoing edge is removed (whichever has smaller weight). The second step does the same with nodes of color $2$, except it ignores edges leading to/ incoming from nodes with color $1$. Observe that at most half the total weight of edges is lost in these two steps. Hence after this step we have $\wt(H) \geq \frac{1}{12} \wt(G_t)$.

    \begin{lemma}
        In Line 26, every connected component in $H$ has diameter at most $10$.
    \end{lemma}
    \begin{proof}
        Orient every edge from $u$ to $\mathsf{out}(u)$. We show that there is no directed path of length at least $6$ in $H$. Because out-degree is at most $1$, on any path in $H$, the direction of the edges can change at most once. Therefore this implies that diameter is at most $10$.

        Suppose, for sake of contradiction, that there existed a directed path of length at least $6$. None of the nodes in the middle of the path can have color $1$, since these nodes must have non-zero in-degree and out-degree. There are four nodes in the middle of the path and can be colored either $2$ or $3$. By pigeon-hole principle, at least one of the nodes must have color $2$ and have non-zero in-degree and out-degree leading to nodes with color $3$. This contradicts the fact that at least one of these edges must have been removed in Line 24.
    \end{proof}

    \begin{lemma}
        In Line 30, $H$ consists of vertex disjoint stars with weight at least $\frac{1}{24} \wt(G_t)$
    \end{lemma}
    \begin{proof}
        Since the diameter of $H$ is $10$, in $O(1)$ rounds, we can compute a spanning forest of $H$. Subsequently either all the even depth or odd depth edges are removed, i.e. diameter of each connected component in $H$ is at most $2$. By the greedy choice at most $\frac{1}{2}$ the weight of $H$ is lost during this procedure.
    \end{proof}

    \def\Dopt{D^{\mathsf{opt}}}
    We now analyze the correctness of the algorithm.
    We have $\wt(G_{t+1}) \leq \frac{23}{24} \wt(G_t)$. The value of $T$ is chosen such that $\wt(G_T) \leq \frac{\epsilon}{234} \wt(G_0)$.

    Let $D$ be the $\CDS$ solution computed independently (and optimally) on the clusters given by $G_T$ and let $\Dopt$ be any optimal solution to the given instance. For a node $u \in G_T$, let $V_u$ be the set of vertices of $G_0$ that were contracted to $u$. Define $\Dopt_u = \Dopt \cap V_u$ and $D_u = D \cap V_u$. Let $W_u$ be the vertices of $G[V_u]$ that have an incident edge of $G_0$ leading to a vertex not in $V_u$. We have that $\Dopt_u \cup W_u$ dominates all clients in $G[V_u]$. Since $D_u$ is an optimal solution, we get,
    \begin{equation*}\begin{split}
        |D_u| &\leq |\Dopt_u \cup W_u| \\
        \Rightarrow \sum\limits_u |D_u| &\leq \sum\limits_u |\Dopt_u \cup W_u| \\
        \Rightarrow |D| &\leq |\Dopt| + 2 |E(G_T)| \\
        \Rightarrow |D| &\leq |\Dopt| + \frac{1}{117} |E(G_0)| \\
        \Rightarrow |D| &\leq (1 + \epsilon) |\Dopt| 
    \end{split} \end{equation*}

    We now analyze the round complexity. We leave it to the reader to verify that Phase 1 can be implemented as a $O(1)$ round distributed protocol (essentially for each line, only $1$ round of communication with neighbors is needed).

    Except Line 17, all other lines in 15-32 can be implemented as $O(1)$ round complexity in the graph $G_t$. Let $d_t$ be the maximum diameter of a clusters given by $G_t$. Any $\LOCAL$ algorithm in $G_t$ can be simulated by $G$ in $d_t$ rounds (collect $G_t$ and simulate). It is already shown that Line 17 takes $O(d_t \log^* d_t)$ time.
    Since $G_{t+1}$ is obtained by contracting stars, we have $d_{t+1} \leq 3 d_t + 2$. This gives $d_t = O(3^t)$. The overall round complexity of Phase 2, is thus, $O(\sum_{t=0}^{T} d_t \log^* d_t) = O(3^T \log^* 3^T) = O(\frac{1}{\epsilon}^c \log^*{\frac{1}{\epsilon}})$ where $c = \log_{24/23} 3$.
\end{proof}

\section{Color Completion Problems}

Consider a graph $G(V,E)$ and a coloring $c : V \mapsto \{1,\dots,k\}$.
The vertex $v$ is {\em properly colored} if each of its neighbors has a different color. 
The classical vertex coloring problem requires deciding if there exists a coloring for which all vertices are properly colored.
When some of the vertices are already assigned a predefined coloring, the resulting recurrent problem is referred to as {\em color completion} ($\PCC$). 
We use the following measures
for evaluating the number of colors used in any valid solution.

\begin{itemize}
\item 
Let $\cP_{pc}$ be the set of colors used by the precolored vertices, and denote $\chi_{pc} = |\cP_{pc}|$.
\item
Let $\cP_{\UN}$ be the set of colors used for the uncolored vertices;
denote $\chi_{\UN} = |\cP_{\UN}|$.
\item
Let $\cP_{new} = \cP_{\UN} \setminus \cP_{pc}$ be the 
{\em new} colors 
used for the uncolored vertices; denote $\chi_{new} = |\cP_{new}|$.
\item
Let $\cP_{all} = \cP_{pc} \cup \cP_{new}$ be the final set of colors 
of all vertices; denote $\chi_{all} = |\cP_{all}|$.
\end{itemize}

For a given instance of $\PCC$, let $\chi^*_{\UN}$  (respectively, $\chi^*_{new}$, $\chi^*_{all}$) be the smallest possible value of $\chi_{\UN}$ (resp., $\chi_{new}$, $\chi_{all}$) over all possible proper color completions of the precoloring.
Additionally, for a given algorithm $\cA$, let $\chi^{\cA}_{\UN}$ (respectively, $\chi^{\cA}_{new}$, $\chi^{\cA}_{all}$) be the value of $\chi_{\UN}$ (resp., $\chi_{new}$, $\chi_{all}$) in the solution computed by $\cA$ for the instance.

The efficiency of an algorithm for $\PCC$ can be measured by two parameters of interest, namely, $\chi_{new}$ and $\chi_{all}$. The difference between them becomes noticeable in instances where the colors in $\cP_{pc}$ are not contiguous. 
We denote by $\PCC_{new}(\chi)$ (resp. $\PCC_{all}(\chi)$) the problem of color completion such that $\chi_{new}$ (resp. $\chi_{all}$) is at most $\chi$.

\subsection{Single Round Color Completion}
\label{sec:singleroundcc}

We first consider what can be done when the online algorithm is restricted to a single round of communication.

\begin{theorem} \label{pcc:preprocess, T=1}
Consider a graph $G$ with maximum degree $\Delta=\Delta(G)$ and chromatic number $\chi=\chi(G)$ with $\Delta > 0$.
We have $\SUPTIME(\PCC_{new}(\chi \cdot \Delta), G) = 1$.
\end{theorem}

\begin{proof}
The algorithm uses the color palette  
$$\cP ~=~ \{(i,j) \mid 1\le i\le \chi, ~~ 1\le j\le \Delta\}.$$
In the preprocessing stage, compute a proper {\em default coloring} $dc$ of the graph using the color palette $\cP^{def} = \{i \mid 1\le i\le \chi\}$, and let each vertex $v$ store its default color $dc(v)$ for future use. These values are not used as colors in the final coloring.

In the recurrent stage,
we are given an arbitrary precoloring $c(w) \in \cP$ for some nodes, and need to complete it to a proper coloring by selecting a color $c(v)$ for each non-precolored node $v$.
(It is assumed that the precoloring itself is proper, i.e., no two precolored neighboring vertices use the same color.)

The algorithm requires a single round of communication. Each precolored node $w$ informs its neighbors about its color $c(w)$. Now consider a non-precolored node $v$. If all neighbors of $v$ are colored, then $v$ chooses a free color from the color palette. As $\chi \cdot \Delta \geq 2 \Delta \geq \Delta + 1$, such a color is guaranteed to exist.

Otherwise, $v$ finds a {\em free color} of the form $(dc(v),j)$ for $1\le j\le\Delta$ satisfying 
$c(w) \ne (i,j)$ for all precolored neighbors $w$ of $v$. The node $v$ then selects $c(v) \gets (dc(v),j)$.

By this algorithm, the color $(i,j)$ selected by $v$ is different from the color of any precolored neighbor of $v$. Also, $(i,j)$ cannot be the selected color of any non-precolored neighbor $w$ of $v$. This is because the default color $dc(w)=i'$ of $w$ satisfies $i'\ne i$, and therefore, the selected color $c(w)$ of $w$ is of the form $(i',k)$ for some $k$, which must differ from $(i,j)$ at least on the first component. Thus, the  coloring $c$ is proper.
\end{proof}

\begin{remark}\label{rmk:w/opreprocess} In the absence of any preprocessing, Linial \cite{DBLP:journals/siamcomp/Linial92} showed that we require \\ $\Omega(\log^* n)$ rounds to color the graph even if it is just a path. 
To complement this, Linial also provides an $O(\log^* n)$ round algorithm that colors the graphs with maximum degree $\Delta$ with\\ $O(\Delta^2)$ colors.

The algorithm works by repeatedly reducing a given proper coloring with $n$ colors to one with at most $\lceil 5 \Delta^2 \log_2 n \rceil$ colors. The same algorithm can be adapted to $\PCC$ with a small change yielding at most $\lceil 23 \Delta^2 \log_2 n \rceil$ new colors. (See Section \ref{sec:pccw/opreprocess}). A consequence of the above is that one can readily adapt existing solutions of graph coloring to color completion. For example the results of Maus \cite{maus2021distributed}, Barenboim et al.\cite{BEG13} can be extended to $\PCC$, with the number of colors used replaced by $\chi_{new}$ and retaining the same round complexities.
\end{remark}

We complement the result of Thm. \ref{pcc:preprocess, T=1} with the following lower bound.
\begin{theorem}
    For every integer $\chi, \Delta$, there exists a graph $G$ with chromatic number $\chi$ and maximum degree $\Delta$ such that for every single round deterministic distributed algorithm $\mathcal{A}$, the total number of colors used by $\mathcal{A}$ over all recurrent instances of $\PCC$ is at least $\chi \cdot (\Delta - \chi + 2)$ even after an arbitrary preprocessing of $G$.
\end{theorem}

\tikzset{mylabel/.style  args={at #1 #2  with #3}{
    postaction={decorate,
    decoration={
      markings,
      mark= at position #1
      with  \node [#2] {#3};
 } } } }

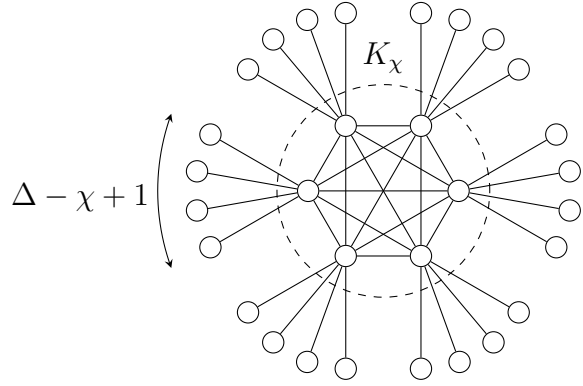
\begin{figure}
    \centering
    \begin{tikzpicture}
        \foreach \i in {1,...,6} {
            \node[draw, circle, inner sep=1mm] (A\i) at ($(\i*60:1)$) {};
        }
        \foreach \i in {1,...,6} {
            \foreach \j in {1,...,4} {
                \node[draw, circle, inner sep=1mm] (B\i\j) at ($(\i*60:1)+(\i*60+\j*20-2.5*20:1.5)$) {};
                
                \draw (B\i\j) -- (A\i);
            }
        }
        
        \foreach \i in {1,...,6} {
            \foreach \j in {\i,...,6} {
                \ifthenelse{\i<\j}{
                \draw (A\i) -- (A\j)}{};
            }
        }
        
        \node[draw, dashed, circle, inner sep=10mm, label=above:$K_{\chi}$] at (0, 0) {};
        \draw [domain=180-20:180+20, mylabel=at 0.7 above left with {$\Delta-\chi+1$},stealth-stealth] plot ({3*cos(\x)}, {3*sin(\x)});
    \end{tikzpicture}
    \caption{Graph whose single round color completion assigns at least $\chi \cdot (\Delta - \chi + 1)$ different colors across all instances. In this example $\chi = 6, \Delta = 9$.}
    \label{fig:cc:singleround:lb}
\end{figure}

\begin{proof}
    The lower bound graph is obtained by taking the clique $K_{\chi}$ and and adding $\Delta - \chi + 1$ different nodes to each node of $K_{\chi}$ (See Figure 
    \ref{fig:cc:singleround:lb}
    ). Let the vertices of the clique be $v_1, v_2 \dots v_{\chi}$ and let $v_{ij}$ denote the $j^{th}$ neighbor of $v_i$ for each $1 \leq i \leq \chi$ and $0 \leq j \leq \Delta-\chi$.
    
    Let $\mathcal{A}$ be any single round deterministic distributed algorithm that solves $\PCC$. We construct $\chi \cdot (\Delta - \chi + 2)$ instances, namely $I_{i, j}$ for each $1 \leq i \leq \chi$ and $0 \leq j \leq \Delta - \chi + 1$ as follows. We define $I_{i, 0}$ to be the instance where none of the nodes are precolored. Let $\col_{i, j}$ be the solution to the instance $I_{i, j}$ given by algorithm $\mathcal{A}$. We construct $I_{i, j}$ (for $j > 0$) from $I_{i, j-1}$ and $\col_{i, j-1}$. $I_{i, j}$ is same as the instance $I_{i, j-1}$ except that vertex $v_{i, j-1}$ is precolored with $\col_{i, j-1}(v_i)$.
    
    We shall now argue that the following $\chi \cdot (\Delta - \chi + 2)$ colors, $\col_{i, j}(v_i)$ for $1 \leq i \leq \chi$ and $0 \leq j \leq \Delta-\chi+1$ are all distinct, which proves the theorem.
    
    Consider $\col_{a, b}(v_a)$ and $\col_{c, d}(v_c)$ for some $1 \leq a < c \leq \chi$ and $0 \leq b, d \leq \Delta - \chi + 1$. To argue that these colors are different, we construct a new instance $I$ wherein $v_{a, j}$ is precolored with $\col_{a, j}(v_a)$ for every $0 \leq j < b$ and $v_{b, k}$ is precolored with $\col_{b, k}(v_b)$ for every $0 \leq k < d$.
    
    Since $\mathcal{A}$ operates in a single round, for node $v_a$, the instance $I$ is indistinguishable from instance $I_{a, b}$. Therefore the color assigned to $a$ by $\mathcal{A}$ for instance $I$ must be $\col_{a, b}(v_a)$.
    
    Similarly, with respect to node $v_c$, the instances $I_{c, d}$ and $I$ are indistinguishable and thus $c$ is assigned $\col_{c, d}(v_c)$ by $\mathcal{A}$ for instance $I$.
    
    Since $v_a, v_c$ are directly connected and $\mathcal{A}$ assigns a proper coloring to $G$ for instance $I$, and $v_a, v_c$ are adjacent in $G$, we have $\col_{a, b}(v_a) \neq \col_{c, d}(v_c)$.
    
    The only pairs left to consider are of the form $\col_{a, b}(v_a)$ and $\col_{a, d}(v_d)$ for some $b < d$. To see that these are different, consider instance $I_{a, d}$. The vertex $v_{a, b}$ is precolored with $\col_{a, b}(v_a)$ and $v_a$ is assigned $\col_{a, d}(v_a)$ by $\mathcal{A}$. Since $v_a, v_{a, b}$ are adjacent, it follows that $\col_{a, b}(v_a) \neq \col_{a, d}(v_a)$. 
\end{proof}

\subsection{\texorpdfstring{$\PCC$}{{\sf CC}} with \texorpdfstring{$\Delta^{1 + \epsilon}$}{Enough} New Colors}
\label{sec:pcctradeoff}
We now describe how the single round Color Completion can be extended for multiple rounds. 

\begin{theorem}\label{pcc:preprocess, T>1}
Consider a graph $G$ with maximum degree $\Delta=\Delta(G)$ and chromatic number $\chi=\chi(G)$ with $\Delta > 0$ and let $k$ be any integer with $1 \leq k \leq \chi$.
We have, 
$$\SUPTIME(\PCC_{new}(\max(\lceil \frac{\chi}{k} \rceil \cdot \Delta, \Delta + 1)), G) \leq k$$
\end{theorem}

\begin{proof}
    The preprocessing stage is same as that of the single round algorithm, where we precompute a proper $\chi$-coloring of the graph. Let $dc(v)$ be the color of $v$. In the recurrent stage, each precolored node $w$ sends its assigned precolor $c(w)$ to all its neighbors during the first round.
    
    Consider the same color palette $\mathcal{P}$ used for the single round color completion, except when $k = \chi$. In case $k = \chi$, add another color $(1, \Delta+1)$ to the palette.
    
    During round $i$ ($1 \leq i \leq k$), nodes $v$ with $dc(v) \equiv i \pmod k$ decide on their colors. If node $v$ has all neighbors precolored, then it chooses any free color of the form either (i) $(1, j)$ for some $1 \leq j \leq \Delta$ or (ii) $(1, \Delta + 1)$ if $\chi = k$ and $(2, 1)$ otherwise. If any neighbor of $v$ is not precolored, then it selects any free color of the form $(\lceil \frac{dc(v)}{k} \rceil, j)$ for some $1 \leq j \leq \Delta$. At least one free color is guaranteed to exist as number of neighboring vertices that have already fixed color before round $i$ is at most $\Delta - 1$. The node finalizes the chosen color as $c(v)$ and if $i < k$, sends $c(v)$ to all its neighbors. 
    
    We now argue that the coloring assigned is proper. It is sufficient to show that whenever a node $v$ adopts a color $c(v)$, $c(v)$ is different from $c(w)$ for all neighbors $w$ of $v$. We always choose $c(v)$ so that it is different from the colors of all neighbors $c(w)$ where $w$ was colored at a previous round. It remains to consider those neighbors of $v$ that are colored in the same round as $v$. Let $w$ be an arbitrary such neighbor. We have $dc(v) \neq dc(w)$ as $dc$ is a proper coloring. Since $dc(w) \equiv dc(v) \pmod{k}$, we must have $\lceil \frac{dc(v)}{k}  \rceil \neq \lceil \frac{dc(w)}{k} \rceil$ and therefore the chosen colors must be different.
\end{proof}

We compare Theorem \ref{pcc:preprocess, T>1} with the algorithm of Maus \cite{maus2021distributed} that colors a graph using $O(\Delta^{1 + \epsilon})$ colors within $\Delta^{\frac{1}{2} - \frac{\epsilon}{2}}$ rounds, i.e. the algorithm uses at most $\frac{c \Delta^2}{k^2}$ colors in $k$ rounds for some constant $c$ and every $k$ with $1 \leq k \leq \sqrt{\Delta}$. 
Comparing the number of colors, the algorithm of Maus uses fewer colors whenever $\sqrt{\Delta} > k > \frac{c \Delta}{\chi}$. 

\subsection{{\sf CC} Without Preprocessing} \label{sec:pccw/opreprocess}
The classical algorithm of Linial \cite{DBLP:journals/siamcomp/Linial92} adopts a coloring in one round with at most $\lceil 5 \Delta^2 \log n \rceil$ colors. The proof is based on the existence of a family of sets that intersect at ``few elements''. The existence of such a family of sets is shown with the help of a probabilistic argument. Specifically, for any given pair of integers $n, \Delta$, there exist $n$ sets $F_1, F_2, \dots F_n$, each a subset of $[m]$ for some integer $m \leq \lceil 5 \Delta^2 \log n \rceil$, that satisfy the following property:
\begin{gather*}
\mathbf{P_0:}     \hskip 2em \forall \{i_0, i_1, i_2, \dots i_{\Delta}\} \subseteq [n], \ \ \left|F_{i_0} \setminus \bigcup_{j=1}^{\Delta} F_{i_j}\right| > 0.
\end{gather*}

The existence of these sets implies a distributed $1$-round
algorithm 
for classical coloring, since a vertex $v$ with a unique identifier $id(v)$ can choose any color from $F_{id(v)} \setminus \bigcup_{u \in \Gamma(v)} F_{id(u)}$. The coloring is proper since the sets satisfy the given property and the maximum color chosen is $m \leq \lceil 5 \Delta^2 \log n \rceil$.

To adapt this algorithm to Color Completion, it is sufficient to modify the property constraint as follows:
\begin{gather*}
\mathbf{P_\Delta:}    \hskip 2em \forall \{i_0, i_1, i_2, \dots i_{\Delta}\} \subseteq [n], \ \ \left|F_{i_0} \setminus \bigcup_{j=1}^{\Delta} F_{i_j}\right| > {\mathbf{\Delta}}.
\end{gather*}

Applying the same probabilistic argument, we can show the following.
\begin{lemma}
For sufficiently large $n$, there exists an integer $m \leq \lceil 23 \Delta^2 \log_2 n \rceil$ and 
sets $F_1, F_2, \dots F_n \subset [m]$, that satisfy property $\mathbf{P_\Delta}$.
\end{lemma}
\begin{proof}
Given $n$ and $m$ as in the lemma, select the sets $F_i$ randomly as follows.
For each integer $x = 1, 2, \dots m$ and each $i = 1, 2, \dots n$, add $x$ to $F_i$ with probability $1/\Delta$.

For a given set $\{i_0, i_1, \dots i_{\Delta}\} \subseteq [n]$, the probability that a particular $x \in [m]$ belongs to $F_{i_0}$ but not the remaining $\Delta$ sets is $\frac{1}{\Delta} \cdot \left(1 - \frac{1}{\Delta}\right)^{\Delta} \geq \frac{1}{4\Delta}$.
Hence, the probability that fewer than $\Delta + 1$ of the elements in $[m]$ belong to $F_{i_0}$ but not the remaining $\Delta$ sets is at most
    $\sum_{j=1}^{\Delta} \binom{m}{j} \left(1 - \frac{1}{4 \Delta}\right)^{m-j}$.
As long as $m > 2\Delta$, the terms are increasing, i.e., $\binom{m}{j+1} x^{m-j-1} > \binom{m}{j} x^{m-j}$. Therefore, we can bound the summation by
\begin{equation}
    \sum_{j=1}^{\Delta} \binom{m}{j} \left(1 - \frac{1}{4 \Delta}\right)^{m-j} \leq \Delta \binom{m}{\Delta} \left( 1 - \frac{1}{4 \Delta}\right)^m.
\end{equation}
Finally, the probability that the chosen sets do not satisfy the property for at least one of the subsets $\{i_0, i_1 \dots i_{\Delta}\}$ is at most
\begin{equation}
    \binom{n}{\Delta + 1} \cdot (\Delta + 1) \cdot \Delta \cdot \binom{m}{\Delta} \cdot (1 - \frac{1}{4 \Delta})^m \leq n^{\Delta + 1} \cdot m^{\Delta} \cdot e^{-\frac{m}{4\Delta}} \cdot \frac{1}{\Delta!}~.
\end{equation}
If the final expression above is strictly less than $1$, then the existence is guaranteed. This occurs whenever $m > 4 \Delta(\Delta + 1) \ln n + 4 \Delta^2 \ln m$. To find such a value of $m$, suppose $c_1 \Delta^2 \ln n < m < c_2 \Delta^2 \ln n$, then $\ln m < \ln{c_2} + \ln{\Delta^2\ln n} < \ln{c_2} + 3\ln n$, using which we can get a weaker (and easily solvable) lower bound for $m$,
\begin{equation*} \begin{split}
    4\Delta(\Delta + 1) \ln n + 4 \Delta^2 \ln m &< 4 \Delta(\Delta + 1) \ln n + 4 \Delta^2 \ln{c_2} + 12 \Delta^2 \ln{n} \\
    &< 20 \Delta^2 \ln n + 4 \Delta^2 \ln c_2
\end{split} \end{equation*}
Therefore, if we can choose $c_1, c_2$ so that $20 + 4 \frac{\ln{c_2}}{\ln n} < c_1 < c_2$ we are done. Considering $n \geq 3$, we can choose any $c_2$ such that $c_2 - (20 + \frac{4 \ln c_2}{\ln 3})$ exceeds $0$. The smallest such value is around $c_2 = 33$, therefore an upper bound on $m$ (and also the maximum number of colors) is at most $33 \Delta^2 \ln n \approx 23 \Delta^2 \log_2{n} $.
\end{proof}

\begin{theorem}
    Color Completion can be solved with $\chi_{new} \leq \chi_{all} \leq \lceil 23 \Delta^2 \log_2{n} \rceil$ colors in one LOCAL round.
\end{theorem}

\subsection{\texorpdfstring{$\PCC$}{{\sf CC}} with fewer than \texorpdfstring{$\Delta+1$}{Delta + 1} colors}
\label{sss: fewer colors}

We next discuss coloring algorithms based on a preprocesing stage, which use fewer than $\Delta+1$ colors when possible. 

\subsubsection{A recurrent algorithm}
~  

Our main result is an algorithm that, for a graph $G$ with chromatic number $\chi$, uses preprocessing, and in the recurrent stage solves any instance of $\PCC$ with at most $\chi$ new colors in $\chi$ rounds. 
The algorithm operates as follows.

\noindent    
\textbf{Preprocessing}.
The preprocessing stage computes a proper-$\chi$ coloring of the graph $G$.
This is stored implicitly, i.e., each node $v$ stores a single color (a positive integer) $dc(v)$. We call this coloring the initial coloring of $G$.

\noindent     
\textbf{Online algorithm}. We call the algorithm the ``priority recoloring'' algorithm. The set of nodes with the same initial coloring form an independent set which implies that nodes belonging to this set may be colored independently. We use the standard greedy algorithm to simultaneously color nodes with the same initial color in a single round. The initial colors are only computed to partition the original set of nodes into $\chi$ independent sets. 

\noindent     
The input of each recurrent instance is a subset $S$ of the nodes that were precolored, i.e., each $v \in S$ has a precolor $c(v)$. For convenience, consider $c(v) = 0$ for all $v \not \in S$.
%
The required output is a color completion of the precoloring: each node $v \not \in S$ outputs a color $c(v) \in \mathbb{N}$ such that the colors assigned to all vertices form
a proper coloring of the graph $G$.

The online algorithm $\cA$ operates as follows.

\begin{itemize}
    \item For $r = 1, 2, \dots \chi$ rounds, do
        \begin{itemize}
            \item If $dc(v) = r$ and $c(v) = 0$ then, $c(v) \gets \min(\mathbb{N} \setminus \Gamma(v))$, where
            $\Gamma(v) = \{c(w) | (w, v) \in E(G)\}$
        \end{itemize}
\end{itemize}

For a given instance of the problem, $\chi^*_{\UN}$  (respectively, $\chi^*_{new}$, $\chi^*_{all}$) is the smallest possible value of $\chi_{\UN}$ (resp., $\chi_{new}$, $\chi_{all}$) over all possible proper color completions of the precoloring, and  $\chi^{\cA}_{\UN}$ (respectively, $\chi^{\cA}_{new}$, $\chi^{\cA}_{all}$) is the value of $\chi_{\UN}$ (resp., $\chi_{new}$, $\chi_{all}$) in the solution computed by the priority algorithm.

\begin{observation}
\label{obs:all=pc+new}
For any coloring, $\chi_{all} = \chi_{pc} + \chi_{new}$.
In particular, 
$\chi^*_{all} = \chi_{pc} + \chi^*_{new}$
and
$\chi^{\cA}_{all} = \chi_{pc} + \chi^{\cA}_{new}$.
\end{observation}

\begin{lemma}
\label{lem:bound chi alg new}
$\chi^{\cA}_{new} \le \chi$.
\end{lemma}

\begin{proof}
For every integer $k\ge 1$, let $\mathbb{N}_k=\{1,\ldots,k\}$.
Let $M=\max\cP_{pc}$, and let $FREE=\mathbb{N}_{M+\chi} \setminus\cP_{pc}$ be the set of free colors (not used in the precoloring) up to $M+\chi$.
Note that the cardinality of the set $FREE$ is at least $\chi$. 
Let $\hat F = \{f_1,\ldots,f_{\chi}\}$ consist of the smallest $\chi$ integers in the set $FREE$.

By induction on $k$ from 1 to $\chi$, one can verify that during iteration $k$ of the algorithm, the colors the algorithm uses for the uncolored vertices of default color $dc(v)=k$ are taken from $FREE \cup \{f_1,\ldots,f_k\}$. Hence $\cP^{\cA}_{\UN} \subseteq \cP_{pc} \cup \hat F$, implying that $\chi^{\cA}_{new} \le |\hat F| = \chi$.
\end{proof}


\begin{theorem}
\label{lem:bounding chi all}
    Consider a graph $G$ with chromatic number $\chi = \chi(G)$. With preprocessing allowed, there exists an algorithm  $\cA$ that can solve an instance of $\PCC$ with $\chi_{all}^{\cA} \leq \chi + \chi^*_{all}-1$ colors and with $\chi_{new}^{\cA} \leq \chi$ in $\chi$ units of time.
\end{theorem}

\subsubsection{Hard examples and negative results}
~  

A natural question is how tight these bounds are. 

Note first that the priority recoloring algorithm does not necessarily yield a good approximation for $\chi_{new}$ (i.e., a bound of the form $\chi^{\cA}_{new} \le \rho\cdot\chi^*_{new}$ for some approximation ratio $\rho$). To see this, consider the example of Fig. \ref{fig:bad-bound-alg-new}. In this example, $\chi^{\cA}_{new} =4$ while $\chi^*_{new}=0$.
%
\begin{figure}[htb]
\centering
\includegraphics[width=3.5in]{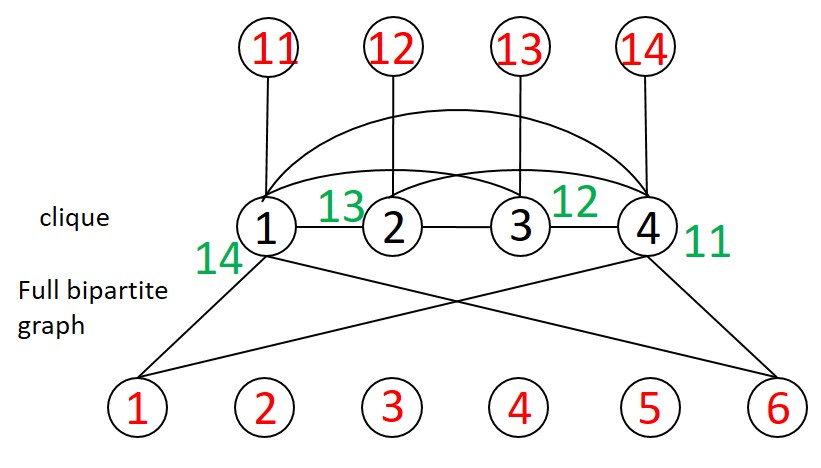}
\caption{Poor approximation for $\chi_{new}$.
Black numbers denote the optimal coloring in the preprocessing stage ($\chi=4$).
The red numbers represent the precoloring ($\chi_{pc}=10$).
The green numbers are a coloring of the clique nodes that optimizes the number of new colors (yielding $\chi^*_{new}=0$).
Note that the priority algorithm will use the new colors 7, 8, 9, 10, so $\chi^{\cA}_{new}=4$.}
\label{fig:bad-bound-alg-new}
\end{figure}
%
In this example, the problem can be attributed in part to the fact that the precoloring uses two non-contiguous blocks of colors, namely, $\{1,\ldots,6\}\cup\{11,\ldots,14\}$.
However, it is possible to construct an example where the priority coloring algorithm performs poorly despite the fact that the precoloring uses a single contiguous block of colors. Consider the graphs constructed recursively as shown in Figures \ref{fig:bad:G2} and \ref{fig:bad:Gk}.

{\bf Initial coloring:} The numbers on the graphs show the initial coloring. Note that the initial colors of the nodes in the cliques $K_{\chi-2}$ are not specified, they must be completed so that they are consistent with those mentioned in the figure. 

{\bf Pre coloring:} The nodes in the cliques ($K_{\chi-2}$) are precolored with colors from $1, \dots \chi - 2$.

For the graph $G_{\chi - 2}$, The priority recoloring algorithm uses $2\chi - 2$ total colors and $\chi - 2$ new colors, however the optimal solution uses only $\chi$ total colors and $2$ new colors. The optimal solution can be obtained by the priority recoloring algorithm if a different initial coloring is chosen, in particular replace color $x$ by color $\chi + 1 - x$ in the same graph and for that initial coloring the priority recoloring algorithm gives an optimal solution.

\begin{figure}[htb]
\vspace{50pt}
\begin{tikzpicture}
\drawgz{0.5}{0.5}{4}{1}
\drawgz{3.5+0.2}{-1}{3}{2}
\drawgz{6.5+0.2}{-2.5}{2}{3}
\drawgz{9+1}{-1}{2}{4}
\draw (c1) -- (c2);
\draw (c1) -- (c4);
\draw (c2) -- (c3);
\draw[rounded corners, dashed] (2.2, -4) rectangle (8.3, 0.5) {};
\node (text1) at (5.3, -3.5) {$G_1$};
\draw[rounded corners, dashed] (8.5, -3) rectangle (11.75, 0.5) {};
\node (text0) at (8.5+1.625, -2.5) {$G_0$};
\node (text1) at (5.3, -4.5) {$G_2$};
\end{tikzpicture}
\caption{Constructing $G_2$ from $G_0, G_1$ (Initial coloring)} 
\label{fig:bad:G2}
\end{figure}
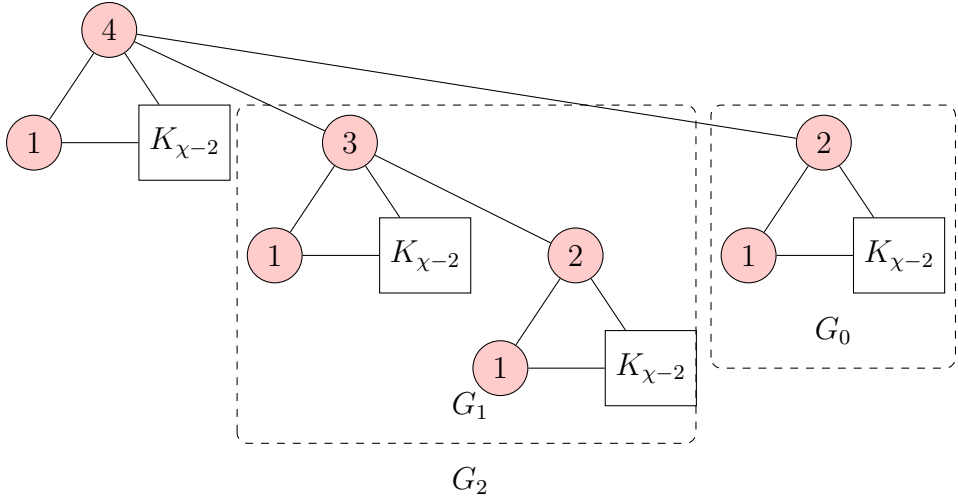

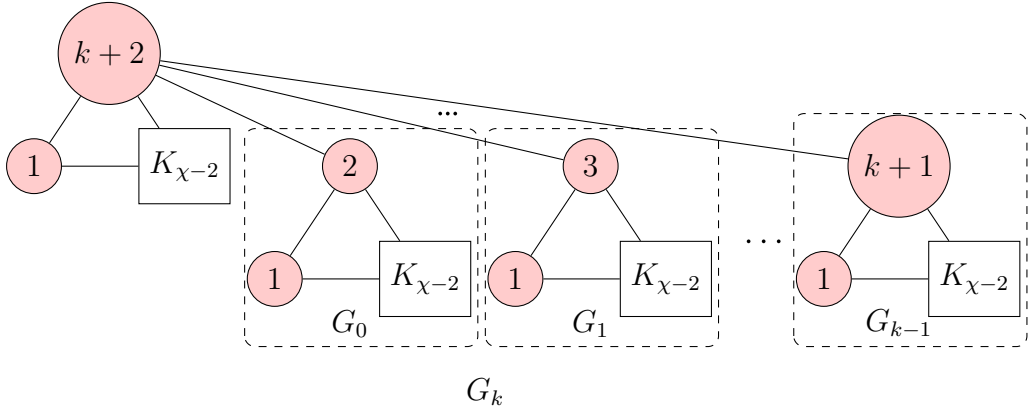
\begin{figure}
\begin{tikzpicture}
\drawgz{0.5}{0.5}{$k+2$}{1}
\drawgz{3.7}{-1}{2}{2}
\drawgz{6.9}{-1}{3}{3}
\node at (9, -1)[circle,fill,inner sep=0.5pt]{};
\node at (9.2, -1)[circle,fill,inner sep=0.5pt]{};
\node at (9.4, -1)[circle,fill,inner sep=0.5pt]{};
\drawgz{11}{-1}{$k+1$}{4};
\draw (c1) -- (c2);
\draw (c1) -- (c3);
\draw (c1) -- (c4);
\node at (6-1.1, 0.7)[circle,fill,inner sep=0.5pt]{};
\node at (6.1-1.1, 0.7)[circle,fill,inner sep=0.5pt]{};
\node at (6.2-1.1, 0.7)[circle,fill,inner sep=0.5pt]{};

\draw[rounded corners, dashed] (3.7-1.4, -1-1.4) rectangle (3.7+1.7, -1+1.5) {};
\node (text0) at (3.7, -1-1.1) {$G_0$};
\draw[rounded corners, dashed] (6.9-1.4, -1-1.4) rectangle (6.9+1.7, -1+1.5) {};
\node (text1) at (6.9, -1-1.1) {$G_1$};
\draw[rounded corners, dashed] (11-1.4, -1-1.4) rectangle (11+1.7, -1+1.7) {};
\node (text2) at (11, -1-1.1) {$G_{k-1}$};
\node (text3) at (5.5, -3) {$G_k$};
\end{tikzpicture}
\caption{Constructing $G_k$ from $G_0, G_1, \dots G_{k-1}$ (Numbers denote Initial coloring) (a) The following precoloring instance is bad : color all the $K_{\chi - 2}$ cliques with colors from $1, 2, \dots \chi-2$ and leave the rest uncolored.}
\label{fig:bad:Gk}
\end{figure}

However, combining Lemma \ref{lem:bound chi alg new} and Obs. \ref{obs:all=pc+new} we get the following.

\begin{corollary}
$\chi^{\cA}_{all} \le \chi_{pc} + \chi$.
\end{corollary}

Since $\chi^*_{all} \ge \max\{\chi_{pc},\chi\}$, we get an approximation of ratio 2 for $\chi_{all}$.

\begin{corollary}
$\chi^{\cA}_{all} \le 2\chi^*_{all}$.
\end{corollary}

\begin{theorem}{(Lower bound for $\chi^{\cA}_{new}$).}
\label{thm:lb-chi alg new}
For every deterministic distributed algorithm $\cA$ that solves $\PCC$ with the guarantee that $\chi^{\cA}_{new} < \chi^{*}_{new} + \chi$, there exists a graph $G$ such that even with preprocessing allowed, there exists an instance of $\PCC$ for which $\cA$ takes $\Omega(D)$ units of time, where $D$ is the diameter of the graph $G$.
\end{theorem}

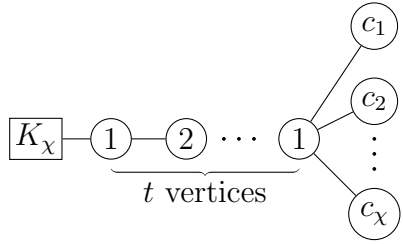
\begin{figure}[htb]
\vspace{50pt}
    \centering
    \begin{tikzpicture}
        \node[draw, rectangle, inner sep = 2pt] (k) at (1, 0) {$K_{\chi}$};
        
        \node[draw, circle, inner sep = 2pt] (a) at (2, 0) {$1$};
        \node[draw, circle, inner sep = 2pt] (b) at (3, 0) {$2$};
        
        \draw (k) -- (a) -- (b);
        
        \node[draw, circle, fill=black, inner sep=0.1pt] at (3.5, 0) {};
        \node[draw, circle, fill=black, inner sep=0.1pt] at (3.7, 0) {};
        \node[draw, circle, fill=black, inner sep=0.1pt] at (3.9, 0) {};
        \node[draw, circle, inner sep=2pt] (f) at (4.5, 0) {$1$};

        \node[draw, circle, inner sep = 2pt] (c) at (5.5, 1.5) {$c_1$};
        \node[draw, circle, inner sep = 2pt] (d) at (5.5, 0.5) {$c_2$};

        \node[draw, circle, fill=black, inner sep=0.1pt] at (5.5, -0) {};
        \node[draw, circle, fill=black, inner sep=0.1pt] at (5.5, -0.2) {};
        \node[draw, circle, fill=black, inner sep=0.1pt] at (5.5, -0.4) {};

        \node[draw, circle, inner sep=2pt] (e) at (5.5, -1) {$c_{\chi}$};
        
        \draw (f) -- (c);
        \draw (f) -- (d);
        \draw (f) -- (e);
        
        \draw [decorate, decoration = {calligraphic brace, mirror}] (2,-0.4) --  (4.5,-0.4);
        \node at (3.25, -0.7) {$t$ vertices};
        
    \end{tikzpicture}
    \caption{Lower bound graph for $\PCC$. $K_{\chi}$ denotes a clique of size $\chi$ and one node of $K_{\chi}$ is connected to the end of the path with $t$ vertices.}
    \label{fig:lower_bound:pcc}
\end{figure}

\begin{proof}
    Consider the graph $G$ shown in Figure \ref{fig:lower_bound:pcc}. The given labels to the nodes denote the precoloring and the none of the nodes in the clique $K_{\chi-1}$ are precolored. The diameter of the graph is $t + 2$.
    
    Consider the set of instances where the precolors $c_1, c_2, \dots c_{\chi}$ are chosen to be distinct integers from the set $S = \{3, 4, \dots 2 \chi + 2\}$. There are in total $\binom{2\chi}{\chi}$ different instance precolorings.
    
    For any deterministic algorithm $\cA$ that runs in $o(t)$ time, the output, consisting of the colors chosen by $\cA$ for the nodes in the clique $K_{\chi}$, must be same for each of the $\binom{2\chi}{\chi}$ instances described above. Let these colors be $\cP_{\UN} = \{\gamma_1, \gamma_2, \dots \gamma_{\chi}\}$. Since $|S| = 2\chi$, $|S \setminus \cP_{\UN}| \geq \chi$ which implies that there exists an instance (colors $c_1, c_2, \dots$ chosen from $S \setminus \cP_{\UN}$) such that $\cP_{c} \cap \cP_{\UN} = \emptyset$ and consequently for that instance, $\chi^{\cA}_{new} = |\cP_{\UN}| = \chi$.
    
    However it is optimal to color the nodes of the clique with the colors $c_1, c_2 \dots c_{\chi}$ which gives $\chi^*_{new} = 0$.
    
    Thus there exists an instance for which $\chi^{\cA} = \chi^*_{new} + \chi$.
\end{proof}

Note that the proof shows also that for any such algorithm $\cA$, there are some instances for which
$\chi^*_{all} = \chi+2$ but $\chi^{\cA}_{all} = 2\chi+2$ and therefore there cannot exist a deterministic CTAS to minimize $\chi_{new}$. 
Randomization also does not help. In the graph constructed above, for any randomized algorithm that takes $o(t)$ rounds, the distribution of the colors assigned to the vertices of the clique $K_{\chi-1}$ must be independent of the values of $c_1, c_2, \dots c_{\chi}$. Furthermore, there must exist a set of $\chi$ colors $T$, such that the probability that the algorithm chooses $T$ is no more than $\frac{1}{\binom{2\chi}{\chi}}$. For the input where $c_1, c_2, \dots c_{\chi}$ are chosen to be from $S \setminus T$, the same bounds for $\chi^{\mathcal{A}}_{all}$ and $\chi^{*}_{all}$ can be achieved. Therefore any algorithm that operates in $o(D)$ rounds and places fewer that $\chi^*_{all} + \chi$ colors cannot succeed with probability more than $\frac{1}{\binom{2\chi}{\chi}}$. This implies the following.

\begin{corollary}
There is no deterministic CTAS for the $\PCC$ problem that minimizes $\chi_{new}$. Furthermore, there is no randomized CTAS that succeeds with any fixed probability. 
\end{corollary}

\noindent
Another implication of 
Thm. \ref{thm:lb-chi alg new}
is that 
without preprocessing, solving $\PCC$ with $\chi^{\cA}_{new} < \chi^{*}_{new} + \chi$ requires time $\Omega(D)$. 

\begin{theorem}\label{pcc:contiguousprecoloring:lb}
    For every integer $\chi \geq 2$ and deterministic algorithm $\mathcal{A}$ that solves $PCC$ with the guarantee that $\chi^{\mathcal{A}}_{new} \leq \chi^*_{new} + 1$, there exists a graph $G$ with chromatic number $\chi$ and a pre-coloring of $G$ for which $\mathcal{A}$ takes $\chi$ units of time, even with arbitrary preprocessing allowed. 
\end{theorem}

\begin{figure}
    \centering
    \begin{tikzpicture}
        \node[draw, circle, inner sep = 2pt] (n11) at (1, 0) {};
        \node[draw, circle, inner sep = 2pt] (n12) at (1, 1) {};
        \node[draw, circle, inner sep = 2pt] (n13) at (1, 2) {};
        \node[draw, circle, inner sep = 2pt] (n14) at (1, 3) {};
        
        \node[draw, circle, inner sep = 2pt] (n21) at (3, 0) {};
        \node[draw, circle, inner sep = 2pt] (n22) at (3, 1) {};
        \node[draw, circle, inner sep = 2pt] (n23) at (3, 2) {};
        \node[draw, circle, inner sep = 2pt] (n24) at (3, 3) {};
        
        \node[draw, circle, inner sep = 2pt] (n31) at (5, 0) {};
        \node[draw, circle, inner sep = 2pt] (n32) at (5, 1) {};
        \node[draw, circle, inner sep = 2pt] (n33) at (5, 2) {};
        \node[draw, circle, inner sep = 2pt] (n34) at (5, 3) {};
        
        \draw (n11) -- (n12);
        \draw (n11) .. controls (0.5,1) and (0.5,1) .. (n13);
        \draw (n11) .. controls (-0,1.5) and (-0,1.5) .. (n14);
        \draw (n12) .. controls (0.5,2) and (0.5,2) .. (n14);
        \draw (n12) -- (n13);
        \draw (n13) -- (n14);
        
        \draw (n21) -- (n22);
        \draw (n21) .. controls (0.5+2,1) and (0.5+2,1) .. (n23);
        \draw (n21) .. controls (-0+2,1.5) and (-0+2,1.5) .. (n24);
        \draw (n22) .. controls (0.5+2,2) and (0.5+2,2) .. (n24);
        \draw (n22) -- (n23);
        \draw (n23) -- (n24);
        
        \draw (n31) -- (n32);
        \draw (n31) .. controls (0.5+4,1) and (0.5+4,1) .. (n33);
        \draw (n31) .. controls (-0+4,1.5) and (-0+4,1.5) .. (n34);
        \draw (n32) .. controls (0.5+4,2) and (0.5+4,2) .. (n34);
        \draw (n32) -- (n33);
        \draw (n33) -- (n34);
        
        \draw[red] (n11) -- (n22);
        \draw[red] (n11) -- (n23);
        \draw[red] (n11) -- (n24);
        
        \draw[red] (n12) -- (n21);
        \draw[red] (n12) -- (n23);
        \draw[red] (n12) -- (n24);
        
        \draw[red] (n13) -- (n21);
        \draw[red] (n13) -- (n22);
        \draw[red] (n13) -- (n24);
        
        \draw[red] (n14) -- (n21);
        \draw[red] (n14) -- (n23);
        \draw[red] (n14) -- (n22);
        
        \draw[red] (n21) -- (n32);
        \draw[red] (n21) -- (n33);
        \draw[red] (n21) -- (n34);
        
        \draw[red] (n22) -- (n31);
        \draw[red] (n22) -- (n33);
        \draw[red] (n22) -- (n34);
        
        \draw[red] (n23) -- (n31);
        \draw[red] (n23) -- (n32);
        \draw[red] (n23) -- (n34);
        
        \draw[red] (n24) -- (n31);
        \draw[red] (n24) -- (n33);
        \draw[red] (n24) -- (n32);
        
        \node at (5+0.3, 1.5)[circle,fill,inner sep=0.5pt]{};
        \node at (5+0.6, 1.5)[circle,fill,inner sep=0.5pt]{};
        \node at (5+0.9, 1.5)[circle,fill,inner sep=0.5pt]{};
        
        \node[draw, circle, inner sep = 2pt] (nl1) at (7, 0) {};
        \node[draw, circle, inner sep = 2pt] (nl2) at (7, 1) {};
        \node[draw, circle, inner sep = 2pt] (nl3) at (7, 2) {};
        \node[draw, circle, inner sep = 2pt] (nl4) at (7, 3) {};
        
        \draw (nl1) -- (nl2);
        \draw (nl1) .. controls (0.5+6,1) and (0.5+6,1) .. (nl3);
        \draw (nl1) .. controls (-0+6, 1.5) and (-0+6,1.5) .. (nl4);
        \draw (nl2) .. controls (0.5+6,2) and (0.5+6,2) .. (nl4);
        \draw (nl2) -- (nl3);
        \draw (nl3) -- (nl4);
        
    \end{tikzpicture}
    \caption{Lower bound graph when $\chi = 4$}
    \label{fig:pcc:contiguouscoloring:lb}
\end{figure}
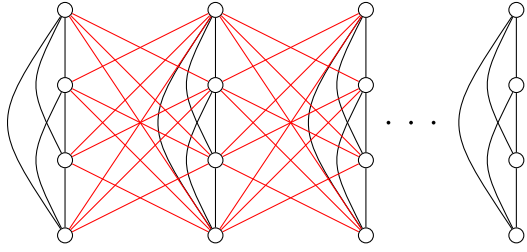

\begin{proof}
    Consider a series of $l$ cliques of size $\chi$. Let $v_{i,j}$ be the $i^{th}$ vertex of the $j^{th}$ clique for $1 \leq i \leq \chi$ and $1 \leq j \leq l$.
    In addition to the $l \cdot \binom{\chi}{2}$ edges between vertices of each clique, add an edge between $v_{i, j}$ and $v_{i+1, k}$ for all $1 \leq i < l$ and $j \neq k$. In particular all pairs of edges between vertices of clique $i$ and clique $i+1$ are connected, except $v_{i, j}$ and $v_{i+1, j}$. See Figure \ref{fig:pcc:contiguouscoloring:lb} for an example with $\chi = 4$.
    
    It is easy to verify that the graph has chromatic number $\chi$. The color assignment $c(v_{i,j}) = j$ is a proper $\chi$-coloring.
    
    The diameter of the graph is $l-1$ and is similar to a path with $l$ vertices except that each vertex is replaced by a clique and between cliques maximum number of edges are added such that chromatic number of the graph remains same.
    
    The only way to color the graph using $\chi$ colors is to assign the vertices $v_{i, 1}, v_{i, 2}, \dots v_{i, l}$ the same color for every $1 \leq i \leq \chi$. 
    Consider a precoloring where only vertices of clique $1$ are colored. In these instances, the color of the vertices in clique $l$ must be same regardless of the pre-colors assigned to vertices of clique $1$. Suppose $c_j$ is the color assigned to $v_{i, l}$, consider the precoloring instance with $c_{pre}(v_{i, 1}) = c_{i\mod{\chi}+1}$ (next color in cyclic order). There is no possible way to complete the coloring in $o(D)$ time without using an additional color.
    
    Suppose the algorithm assigns color $\chi + 1$, then between any two adjacent cliques ($i, i+1$), there can be at most one $j$ such that $c(v_{i,j}) \neq c(v_{i+1, j})$. Therefore at least one of the cliques $\chi+1, \chi+2, \dots l$ must have a different output when the input is changed. However this cannot occur if the algorithm takes less than $\chi$ units of time.
\end{proof}

\section{Recurrent Locally Checkable Labellings (\LCL)}

Locally Checkable Labellings (LCL) were first proposed by Naor and Stockmeyer \cite{DBLP:journals/siamcomp/NaorS95}. Informally, an LCL problem on a graph $G$ asks for an assignment $\Gamma_{out}$ of \emph{labels}, to the vertices of $G$ that satisfy a set of rules that are verifiable ``locally". These are problems whose solutions can be verified by an $O(1)$ round distributed algorithm in the $\LOCAL$ model. Whenever the solution is incorrect, \textit{at least} one of the nodes in the graph identifies so (not necessarily all of them).

\begin{definition}[Locally Checkable Labellings (LCL)]
    An LCL problem for a graph $G$ is described by a 5-tuple $(r, \Sigma_{in}, \Sigma_{out}, \Gamma_{in}, \mathcal{C})$ where 
    \begin{enumerate}
        \item $\Sigma_{in}$ is a set of \emph{input labels},
        \item $\Gamma_{in} : V(G) \rightarrow \Sigma_{in}$, is an \emph{assignment}  of input labels to each vertex of $G$
        \item $\Sigma_{out}$ is a set of \emph{output labels}
        \item $\mathcal{C}$ is a set of \emph{rules}. Each element of $\mathcal{C}$ is a labelled centered graph $H$ with a designated center $w\in V(H)$, and a labelling $\Gamma: V(H) \mapsto \Sigma_{in} \times \Sigma_{out}$. 
        The distance of every node in $H$ from $w$ is at most $r$.
        
        For a given vertex $u \in V(G)$, let $G_r(u)$ be the graph induced by vertices $v$ of $G$ that are at a distance at most $r$ from $u$.
        
        A given labelling $\Gamma_{out} : V(G) \rightarrow \Sigma_{out}$ is \emph{valid} if and only if for every vertex $u \in V(G)$, there is a graph $H \in \mathcal{C}$ and an isomorphism $\phi : V(G_r(u)) \rightarrow V(H)$ such that,
        \begin{itemize}
            \item $\phi(u)$ is the designated center of $H$
            \item $(\Gamma_{in}(u), \Gamma_{out}(u)) = \Gamma(\phi(u))$
        \end{itemize}
    \end{enumerate}
\end{definition}

Problems such as computing (an arbitrary) Dominating Set, Vertex Cover, Maximal Matching, $\Delta + 1$ Coloring can be represented as LCLs. The examples mentioned previously do not require input labels (i.e., we can construct LCL's where every vertex has the same input label). Problems such as finding 
a
client dominating set or a color completion (i.e., variants of the classical problems with PFO or PCS instances) can also be captured by the above definition, however they \textit{crucially} require input labels, i.e. $|\Sigma_{in}| > 1$ for these LCL's.

To realise the Client Dominating Set as an LCL, consider $\Sigma_{in}$ to be $\{\textsf{client}, \textsf{non-client}\}$ and $\Sigma_{out} = \{\textsf{server}, \textsf{non-server}\}$. The input labelling $\Gamma^{in}$, assigns the input labels accordingly as per the client set $C$ given by the $\CDS$ instance. The set of rules $\mathcal{C}$ consists of all centered graphs with radius $1$ and degree at most $\Delta(G)$ wherein one of the following holds: (i) the center is labelled a $\textsf{server}$, (ii) one of the neighbors of the center is labelled as a $\textsf{server}$ or (iii) the center has input label \textsf{non-client}. Restricting $\Sigma_{in} = \{\textsf{client}\}$ captures the classical Dominating Set problem. Note that LCL's are often \textbf{not} optimisation problems, i.e. we often can't minimize/maximize any set of labels as such problems are often not locally verifiable.

\subsection{Subgraph LCL's without Input Labels on Paths}
\label{sec:subgraphlclw/oinput}
In this section we consider a subset of recurrent LCL's, named \emph{subgraph LCL's without input labels}, which were studied by Foerster et al. \cite{foerster2019power}.
In subgraph LCL's, the online instances ask for a valid labelling for some (edge induced) subgraph of the given graph $G$. 
This class of LCL's is easier to solve, but already captures several classical problems, such as finding a dominating set, maximal matching, maximal independent set, $(k, l)$-ruling sets etc.

We consider subgraph LCL on a path $P_n$. 
Before getting to the solution, we first remark that one may consider without loss of generality only LCL's with radius $1$. Given an LCL problem of radius $r$, one may construct an equivalent LCL with radius $1$ at the cost of increasing the output label size and the set of rules.

From a prior work (Theorem 3 in Foerster et al. \cite{chang2019exponential}), we may infer that if the round complexity of $\Pi$ in the $\LOCAL$ model is $o(n)$, then it must be $O(1)$ in the $\SUPPORTED$ model. This result is non-constructive, i.e., it argues that given a $o(n)$ round distributed algorithm, one can transform it into an $O(1)$ round algorithm. Additionally, it does not help categorize LCL problems that are $\Theta(n)$ in the $\LOCAL$ model. Some LCL problems (such as $2$-coloring) are $\Theta(n)$ in the $\LOCAL$ model, but clearly $O(1)$ in the $\SUPPORTED$ model. One can also construct LCL's that remain $\Theta(n)$ in the $\SUPPORTED$ model. Furthermore, the proof offers no insight about the additional amount of memory per node that is needed for the preprocessing stage. The following theorem addresses the above questions. Note that as done in prior work, we treat the size of the description of $\Pi$ as constant in the round complexity (in particular, $|\Sigma_{out}|$ and $|\Sigma_{in}|$ are constants).

\begin{theorem}
    Let $\Pi$ be a subgraph LCL with $|\Sigma_{in}| = 1$ and let $P_n$ be a path on $n$ vertices, then 
    \begin{itemize}
        \item $\SUPTIME(\Pi, P_n)$ is either $\Theta(1)$ or $\Theta(n)$
        \item $\SUPSPACE(\Pi, P_n)$ is $O(1)$
        \item $\SUPTIME(\Pi, P_n)$ and an optimal solution for $\Pi$ can be found in time polynomial in size of \ $\Pi$ by a centralized algorithm.
    \end{itemize}
\end{theorem}

\begin{proof}
    As remarked earlier, we may assume that the radius $r$ for the LCL problem is $1$. Therefore, on a path we can represent $\mathcal{C}$ as consisting of centered paths of length 1, 2 or 3, whose (ordered) label sets form a subset of $\Sigma_{out} \cup \Sigma^2_{out} \cup \Sigma^3_{out}$ (recall that $|\Sigma_{in}| = 1$ and can be ignored).
    Note that the tuples in $\mathcal{C}$ are ordered, in particular $(a, b)$ is different from $(b, a)$.
    For tuples of length $2$, we assume the first element is the label of the center and for tuples of size $3$, we assume that the middle element is the center. For example, $(a, b) \in \mathcal{C}$ represents a path of length $2$ with the center labeled $a$. Similarly $(a, b, c)$ represents a path of length $3$ with the center labelled $b$ and the endpoint vertices labelled $a$ and $c$. 
   
    Construct a directed graph $G_d$ defined as follows. 
    Its vertex set is $V(G_d)=\Sigma^2_{out}$, and $E(G_d)$ contains a directed edge from $(a, b)$ to $(b, c)$ if and only if $(a, b, c) \in \mathcal{C}$. 
    Define the \emph{starting} and \emph{terminal} vertices of $G_d$ as
    \begin{eqnarray*}
    S &=& \{(a,b)\in \Sigma^2_{out} \mid (a,b)\in\mathcal{C}\},
    \\
    T &=& \{(a,b)\in \Sigma^2_{out} \mid (b,a)\in\mathcal{C}\}.
    \end{eqnarray*}
  
    The key observation underlying our proof is that finding a solution to the LCL problem on a path of length $n$ is equivalent to finding a walk in $G_d$ of length $n$ that begins at a starting vertex and ends at a terminal vertex.
    
    \begin{claim}\label{clm:lcl:walk}
        For every path $P_n$ with $n \geq 3$, an assignment of output labels $(s_0, s_1, \dots s_{n-1})$ is valid if and only if $(s_0, s_1), (s_1, s_2) \dots (s_{n-2}, s_{n-1})$ is a walk in $G_d$ that begins at a starting vertex in $S$ and ends at a terminal vertex in $T$.
    \end{claim}
    \begin{proof}
        ($\Rightarrow$) By correctness of the solution we must have $(s_0, s_1), (s_{n-1}, s_{n-2}) \in \mathcal{C}$. By definition, $(s_0, s_1) \in S$ and $(s_{n-2}, s_{n-1}) \in T$. 
        
        By correctness of the LCL we have, $(s_{i-1}, s_{i}, s_{i+1}) \in \mathcal{C}$ and therefore there exists an edge between $(s_{i-1}, s_i)$ and $(s_i, s_{i+1})$ for every $1 < i < n-1$. Hence the given sequence represents a walk in $G_d$.
        
        ($\Leftarrow$) As starting and terminal vertices are in $S$ and $T$, respectively, we have $(s_0, s_1) \in \mathcal{C}$ and ~$(s_{n-1}, s_{n-2}) \in \mathcal{C}$. Therefore the rules are satisfied for the ends of the path. For the intermediate vertices we have that there is an edge from $(s_{i-1}, s_i)$ to $(s_i, s_{i + 1})$ for every $1 < i < n-1$, and therefore $(s_{i-1}, s_i, s_{i+1}) \in \mathcal{C}$. Hence the sequence is a valid assignment of labels.
    \end{proof}
    
    \begin{definition}[walkspan]
        Given a directed graph $G_d$ and two vertices $u, w \in V(G_d)$, define ~$\mathrm{walkspan}(u, w)$ as the set of lengths of walks in $G_d$ that start at $u$ and end at $w$. We extend the definition to subsets $U, W \subseteq V(G_d)$ in the natural way, i.e., $\mathrm{walkspan}(U, W) = \bigcup_{u \in U, w \in W} \mathrm{walkspan}(u, w)$.
    \end{definition}
    
Let $\alpha=|\Sigma_{out}|^2$. For a set of integers $S$ and a positive integer $k$, let $S/k = \{j \pmod k \mid j \in S\}$.

    \begin{lemma}
        If ~$\SUPTIME(\Pi, P_n) = o(n)$ then $G_d$ contains a cycle $C$ and a vertex $v \in C$ such that 
        \begin{equation}\label{eqn:criteria}
            \mathrm{walkspan}(S, v) / |C| ~=~ \mathrm{walkspan}(v, T) / |C| ~=~ \{0, 1, \dots |C|-1\}.
        \end{equation}
   \end{lemma}
   
   \begin{proof}
   Let $\mathcal{A}$ be any distributed algorithm for the online phase  that solves $\Pi$ in $o(n)$ (recall that $\mathcal{A}$ can use any information obtained out of an arbitrary preprocessing phase).
   Let the given path be $P_n = (v_1, v_2, \dots v_n)$, ordered from its left end to its right end. We assume $n>6\alpha$.
    Consider the subpath ~$Q = (v_{n/2-\alpha/2}, ..., v_{n/2 + \alpha/2 + 1})$ i.e., a path of length at least $\alpha + 1$, around the center of $P$. 
   Now construct $\alpha$ instances for the online phase, namely, $I_1, I_2, \dots I_{\alpha}$, where 
   $I_i=(v_i, v_{i+1}, \dots v_{n-i+1})$ for $1 \le i \le \alpha$, namely, each $I_i$ is obtained from $P_n$ by removing the $i-1$ first and last vertices.
   
   Note that since $n>6\alpha$, the first (respectively, last) vertex of $Q$ is at distance $\Omega(n)$ from $v_{\alpha}$ (resp., $v_{n-\alpha+1}$). Hence, each of the instances $I_i$ fully contains the subpath $Q$, and moreover, its start segment (from $v_i$ to the first vertex of $Q$) and end segment (from the last vertex of $Q$ to $v_{n-i+1}$) are of length $\Omega(n)$.
    This implies that during the execution of the online algorithm on any given recurrent instance $I_i$, the vertices in $Q$ cannot distinguish between any of the instances constructed above (i.e., they will see exactly the same inputs, and consequently perform exactly the same steps, on each of these instances). Consequently, for every vertex $v_j$ in $Q$, the output of $\mathcal{A}$ is the same for every instance $I_i$. Let the output be $\bar\psi = (s_0, s_1, \dots s_{n-1})$.
    
   As $|Q| > \alpha$, there exists a subpath of $Q$, say $\bar Q$, whose assigned labels 
   $s_t, s_{t+1}, \dots s_{j-1}, s_j$ form a simple cycle, i.e., such that $s_{j-1}=s_t$ and $s_j = s_{t+1}$. 
   By the correctness of these labels, we have that $(s_t, s_{t+1}), (s_{t+1}, s_{t+2}) \dots (s_{j-1}, s_{j})$ is a cycle in $G_d$. Denote this cycle by $C$ and let the first vertex of $\bar Q$ be $v_{\ell}$. We show that $C$ and $(s_t, s_{t+1})$ are the desired cycle and vertex satisfying the properties of the lemma.
    
    Note that in all the instances $I_1, I_2, \dots I_{\alpha}$, the labels of the vertices $v_{\ell}, v_{\ell+1}$ assigned by $\mathcal{A}$ remain the same (i.e., $s_t, s_{t+1}$ respectively). 
    Consider instance $I_i=(v_i, v_{i+1}, \dots v_{n-i+1})$.
    Let the assigned labels by $\mathcal{A}$ to this path be $\psi=$ $(s'_1, s'_2, \dots s'_{\ell-i+1}, s'_{\ell-i+2} \dots s'_{n-i+1})$. We have $s'_{\ell-i+1} = s_t$ and $s'_{\ell-i+2} = s_{t+1}$ (Here $s'_{\ell-i+1}, s'_{\ell-i+2}$ are the labels of $v_{\ell}, v_{\ell+1}$ respectively). Note that $\psi$ is valid. Therefore, by Claim \ref{clm:lcl:walk}, $(s'_1, s'_2) \in S$ and $(s'_1, s'_2), (s'_2, s'_3), \dots (s'_{\ell-i+1}, s'_{\ell-i+2})$ is a walk of length $\ell-i+1$ in $G_d$ that ends at $(s'_{\ell-i+1}, s'_{\ell-i+2})=(s_t, s_{t+1})$. It follows that for every $i = 1, 2, \dots \alpha$, there exists a walk that (i) starts from some vertex in $S$, (ii) ends at vertex $(s_t, s_{t+1})$ and (iii) is of length $\ell-i+1$. We have shown that $\mathrm{walkspan}(S, (s_t, s_{t+1}))$ contains $\alpha \geq |C|$ contiguous integers and hence $\mathrm{walkspan}(S, (s_t, s_{t+1})) / |C| = \{0, 1, \dots |C|-1\}$. By a symmetric argument we can show that $\mathrm{walkspan}((s_i, s_{t+1}), T) / |C| = \{0, 1, \dots |C|-1\}$.
   \end{proof}

    \begin{lemma}
    If $G_d$ contains a cycle $C$ and a vertex $v \in C$ that satisfies Equation \eqref{eqn:criteria}, then $\Pi$ is solvable in $O(\alpha^2) = O(|\Sigma|^4)$ rounds in the \SUPPORTED model.
   \end{lemma}
   
   \begin{proof}
   We first compute the shortest length walks of each congruence class in $\mathrm{walkspan}(S, v)$ and $\mathrm{walkspan}(v, T)$ modulo $|C|$. We show that the shortest such walk has length at most $2 \alpha^2$.
        
 Consider any walk $W = (w_0, w_1, \dots)$ in $G_d$.
 Decompose the walk into an alternating sequence of simple paths 
 and cycles, i.e., $W = P_0 \circ C_0 \circ P_1 \circ C_1 \dots$. Such a decomposition can be obtained by finding the smallest prefix of the walk that contains a simple cycle, say $P_0 \circ W_0$. Remove the vertices of $P_0 \circ W_0$ except the last vertex, and repeat recursively for the remaining part of the walk.
 We would like to shorten the walk $W$, while maintaining two invariants: (i) the remainder $|W|\pmod{|C|}$ obtained when the length of the walk is divided by $|C|$, and (ii) the fact that $W$ starts at a vertex from $S$ and ends at a vertex in $T$. We first observe that removing cycles in the walk does not affect invariant (ii). To achieve (i) we use the following well known number-theoretic fact.
        
    \begin{proposition}
       For any sequence of $n$ (not necessarily distinct) integers $a_1, a_2, \dots a_n$, there exists a subset of these integers whose sum is divisible by $n$.
    \end{proposition}
        
    \begin{proof}
       Define $s_i = (a_1 + a_2 + \dots + a_i) \pmod n$ for $i = 1, 2, \dots n$ and define $s_0 = 0$. By the pigeon-hole principle, there exist $0 \leq i < j \leq n$ such that $s_i = s_j$. The desired set is $\{a_{i+1}, a_{i+2}, \dots, a_j\}$.
    \end{proof}
        
    Apply the following shortening process to $W$. While there are at least $|C|$
    cycles in the walk decomposition, choose any subset of the cycles whose total length is divisible by $|C|$ and remove them. At the end of this process, we are left with a sequence of simple paths and cycles with at most $|C| - 1 < \alpha$ cycles and at most $|C| \leq \alpha$ paths.
    Each simple cycle and path contains at most $\alpha$ vertices and therefore the length of the final shortened walk $W$ is at most $2 \alpha^2$.
    
    We are now ready to describe the distributed recurrent algorithm for the solving $\Pi$, consisting of a preprocessing stage and an online procedure.
    
    \textbf{Preprocessing Stage.} In the preprocessing phase, we first compute a candidate cycle and vertex pair $C, v$ satisfying Equation \eqref{eqn:criteria}. Since $\Pi$ is global knowledge, $C, v$ can be reconstructed by each node in the online stage, as long as they use the same deterministic algorithm to find it. We only require the length of the cycle, $|C|$. Split the path into blocks of size exactly $|C|$, except possibly the last block. Color each node using two colors $0, 1$ such that two adjacent nodes have the same color if and only if they belong to the same block. Let this coloring be $\psi$. In addition to the above decomposition, we also orient the edges such that every node has outdegree at most $1$. This gives a consistent left to right orientation to the nodes of the path. We require only $1$ bit to be stored in each node, namely which of its neighbors has the outgoing edge. In total we have only two bits of information given to each node during the preprocessing stage, one bit for orientation and another bit for the block decomposition.
    
    \textbf{Online Stage.} Each node computes a candidate $C, v$ that satisfies Equation \eqref{eqn:criteria} using the same deterministic algorithm.  We also compute the shortest length walks $L_1, L_2, \dots L_{|C|}$ from a vertex in $S$ to $v$ and the walks $R_1, R_2 \dots R_{|C|}$ from $v$ to a vertex in $T$ such that $|L_i| \equiv |R_i| \equiv i \pmod {|C|}$. We discuss later how all of the above information can be obtained by centralized algorithms that run in time polynomial in $|\Sigma_{in}|$ (This bound does not affect round complexity, but shows that nodes only perform computation that is polynomial in $|\Sigma_{in}|$). 
    
    Let $I$ be the online instance which is a set of subpaths of the path $P$. We solve each subpath independently. Let $P_s$ be a subpath in $I$. We may assume that $P_s$ has at least $\alpha^2 + 2\alpha$ nodes, otherwise the instance can be solved by a single node that collects subpath $P_s$. 
    
    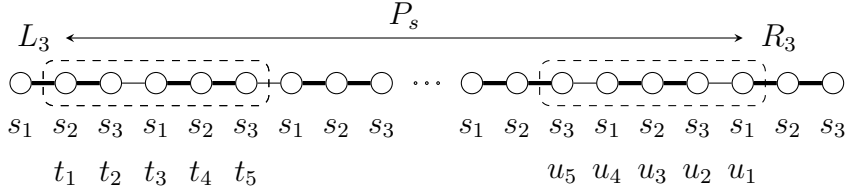
\begin{figure}
    \centering
    \begin{tikzpicture}[scale=0.6]
    
        \foreach \i in {1,...,9} {
            \tikzmath{\j=\i-1;}
            \pgfmathparse{int(Mod(\j, 3))};
            \tikzmath{\m=\pgfmathresult;}
            \node[draw, circle, inner sep=1mm] (A\i) at ($(\i,0)$) {};
            \pgfmathsetmacro{\thickness}{ifthenelse(\pgfmathresult==0, "", "ultra thick")};
            \ifthenelse{\i>1}{\draw[\thickness] (A\i) -- (A\j)}{};
            
            \pgfmathparse{int(\pgfmathresult+1)};
            \node at ($(\i,-1)$) {$s_{\pgfmathresult}$};
        }
        
        \foreach \i in {1,...,3} {
            \node[draw, circle, inner sep = 0.1mm] at ($(\i*0.25+9.5, 0)$) {};
        }
        
        \foreach \i in {10,...,18} {
            \tikzmath{\j=\i+1;}
            \tikzmath{\jj=\i-1;}
            \node[draw, circle, inner sep=1mm] (A\i) at ($(\j,0)$) {};
            \pgfmathparse{int(Mod(\jj-9, 3))};
            \pgfmathsetmacro{\thickness}{ifthenelse(\pgfmathresult==0, "", "ultra thick")};
            \ifthenelse{\i>10}{\draw[\thickness] (A\i) -- (A\jj)}{};
            
            \pgfmathparse{int(\pgfmathresult+1)};
            \node at ($(\j,-1)$) {$s_{\pgfmathresult}$};
        }

        \draw[stealth-stealth] (2, 1) -- (17, 1);
        \node at (9.5,1.5) {$P_s$};
        
        \foreach \i in {1,...,5} {
            \node at ($(\i+1, -2)$) {$t_{\i}$};
        }
        
        \foreach \i in {1,...,5} {
            \node at ($(18-\i, -2)$) {$u_{\i}$};
        }
        
        \draw[rounded corners, dashed] (2-0.5,0.5) rectangle (6+0.5,-0.5);
        \node at (1.3, 1) {$L_3$};

        \draw[rounded corners, dashed] (2-0.5,0.5) rectangle (6+0.5,-0.5);
        \draw[rounded corners, dashed] (13-0.5,0.5) rectangle (17+0.5,-0.5);
        \node at (17.8, 1) {$R_3$};
    \end{tikzpicture}
    \caption{Given path $P_n$ decomposed into blocks of size $k=3$. A subpath $P_s$ chosen for some online instance. Labels $s_i$ are obtained from the cycle $C$ of the corresponding graph $G_d$.}
    \label{fig:lcl:alg}
\end{figure}
    
    Decompose the subpath $P_s$ into blocks by removing edges whose ends (say $u, v$) have different colors ($\psi(u) \neq \psi(v)$). This decomposes $P_s$ into blocks of size exactly $k = |C|$, except possibly the first and last blocks (i.e. those containing the ends of the subpath). Each block is also oriented from left to right (using the orientation remembered from the preprocessing phase). Let the labels of the cycle $C$ be $(s_1, s_2), (s_2, s_3) \dots (s_k, s_1), (s_1, s_2)$ with $(s_1, s_2)$ denoting the vertex $v$. Label the $i^{th}$ vertex of the block (numbered from left) with $s_i$. This is already ``almost" a valid labelling for $P_s$, because except for a ends of the subpath $P_s$, all nodes see a graph from $\mathcal{C}$ in their local $1$-neighborhood. Let the number of vertices in the first and last blocks of $P_s$ be $a', b'$ respectively and let $a, b$ be integers such that $1 \leq a, b, \leq k$ and $a \equiv a'+1 \pmod{k}$, $b \equiv b'-1 \pmod{k}$. $L_a = \{(t_1, t_2) (t_2, t_3) \dots (t_{|L_a|-1}, s_1) (s_1, s_2)\}$. Replace the labels of the first $|L_a|-1$ nodes of $P_s$ with $t_1, t_2, \dots t_{|L_a|-1}$. Similarly replace the labels of the last $|R_b|-1$ vertices with the labels obtained from the walk $R_b$. Both can be done distributively in $|L_a| + |R_b| = O(\alpha^2)$ rounds by having the ends of the subpath $P_s$ relay this information to the nodes. As the length of the path is at least $4 \alpha^2 + 2\alpha$, the first $|L_a|-1$ and last $|R_b|-1$ vertices of subpath $P_s$ are disjoint. The resulting labelling traces a walk in $G_d$ from a vertex in $S$ to a vertex in $T$ and by Claim \ref{clm:lcl:walk} is a valid labelling for the LCL $\Pi$.
    
    Figure \ref{fig:lcl:alg} shows an example of an LCL whose cycle, $C = \{(s_1, s_2), (s_2, s_3), (s_3, s_1)\}$, contains three vertices. The online subpath $P_s$ is such that the precomputed block decomposition decomposes the first and last blocks into sizes $a'=2, b'=1$ respectively. We have $a=3, b=3$ to be the desired walk lengths. In this example $L_3, R_3$ both are walks on $6$ vertices. $L_3 = \{(t_1, t_2), (t_2, t_3), \dots (t_5, s_1), (s_1, s_2)\}$. Similarly $R_3$ is the walk $\{(s_1, s_2), (s_2, u_5), (u_5, u_4) \dots (u_2, u_1)\}$.
    
We conclude with justifying the algorithm and its time complexity analysis.
First, notice that we do not need to precompute the cycle $C$ in the preprocessing stage. Given the description of $\Pi$, we can verify in the online execution if there exists a cycle $C$ and vertex $v$ satisfying Equation \eqref{eqn:criteria}. This can be done in a single round (in which the number of computational steps performed locally at each vertex is polynomial in $|\Sigma_{in}|$),
as $C$ has $\alpha = |\Sigma|^2$ vertices.
To compute (online) the desired walks (or establish that they do not exist), note that we only need to consider walks of length at most $\alpha^2 = |\Sigma|^4$. 
   Consider the graph $G_w$ with vertices $(u, j)$ for $u \in G_d$ and $0 \leq j < \alpha^2$. Draw an edge from $(u_1, j)$ to $(u_2, j+1)$ if there exists an edge $(u_1, u_2)$ in $G_d$. Add a vertex ${\sf src}$ to the graph and an edge from ${\sf src}$ to $(u_1, 0)$ if $u_1 \in S$. Determining if there exists a walk in the graph from a vertex in $S$ to $v$ of length $j$ is equivalent to finding a path in $G_w$ from ${\sf src}$ to $(v, j)$. 
   
   Finally to test if a cycle of length $c$ exists, find the smallest $j \leq \alpha$ such that (i) there exists a walk from $S$ to $v$ and $v$ to $T$ of length $0, 1, \dots j-1 \pmod j$, and (ii) there exists a walk from $v$ to $v$ of length $j$. This can again be determined in a single round, spending a polynomial number of computation steps locally at each vertex.
   \end{proof}
   
   We have thus shown that $\SUPTIME(\Pi, P_n)$ is either $\Theta(1)$ or $\Theta(n)$, and $\SUPSPACE(\Pi, P_n) = O(1)$, i.e., to achieve optimal round complexity, it is sufficient to remember $O(1)$ bits. Precomputing more information does not help to reduce round complexity of this particular problem.
   
   \textbf{Remark.} The bound $O(|\Sigma|^4)$ on the round complexity cannot be improved further in general, i.e., there are LCL's that require these many rounds. We construct them by first constructing the graph $G_d$. Consider the graph $G_d$ which consists of two directed cycles that share exactly one common vertex and both cycles have length $\Omega(|\Sigma|^2)$ and are prime. Exactly one vertex is contained in $S$ and this vertex is also in $T$ and this vertex is NOT the common vertex of the two cycles. In order to realise all congruence classes, we need a walk of size at least $p_1 p_2 - \mathsf{max}(p_1, p_2)$ which is $\Omega(|\Sigma|^4)$.
\end{proof}

\newcommand{\inpl}{\Sigma_{in}}
\newcommand{\outpl}{\Sigma_{out}}

\subsection{Recurrent LCL's on Paths}
\label{sec:reclcl:paths}

In this section we look at a broader class of recurrent \emph{in-labeled LCL's}, namely, LCL's augmented with input labels, wherein the online instances specify different input labels. The set of rules $\mathcal{C}$ and the set of output labels $\outpl$ remain the same across instances. The only components of the input that vary across instances are the input labels $\Gamma^{in}$. Subgraph LCL's, studied in Sect. \ref{sec:subgraphlclw/oinput}, can be represented in this framework by encoding the adjacent edges that are present in the input labels for each vertex, hence in-labeled LCL's are a generalisation of subgraph LCL's. Problems such as finding a Client Dominating Set, Color Completion, Maximal Matching and in general variants of classical local problems with PFO and / or PCS fall into this category.

We show that even for these instances, finding the optimal round complexity is either $\Theta(1)$ or $\Theta(n)$, thus extending the distributed speed up theorem in Foerster et al. \cite{foerster2019power}. However, so far we were unable to find a characterization as obtained in the previous subsection. Therefore, we are left with a couple of intriguing open questions. First, we do not know any bounds on the constant of the running time in terms of the size of the LCL $\Pi$, namely,  $|\inpl|+ |\outpl|$. Second, we do not know if it is possible to decide the online round complexity in polynomial time given the description of the LCL. 

\begin{theorem}
    Let $\Pi$ be a recurrent LCL problem whose online instances differ only in the assignment of input labels. Then $\SUPTIME(\Pi, P_n)$ is either $\Theta(1)$ or $\Theta(n)$ and $\SUPSPACE(\Pi, P_n) = O(1)$.
\end{theorem}

\newcommand{\preprocessinf}{{\sf Inf}_{pre}}

\begin{proof}
    Suppose $\SUPTIME(\Pi, P_n) = T(n) = o(n)$. Let $\mathcal{A}(n)$ be the online algorithm and denote the information obtained during the preprocessing phase for a path of length $n$ by $\preprocessinf(n)$.
    
    Using $\mathcal{A}(n)$ and $\preprocessinf(n)$, we design a modified preprocessing phase, yielding the information $\preprocessinf'(n)$, and a modified online algorithm $\mathcal{A'}(n)$ with improved online round complexity.
    
    Since $T(n) = o(n)$, there must exist a constant $n_0$ such that $4T(n_0) < n_0$.
    We only describe the procedure for paths whose lengths are more than $2n_0$. The remaining cases can be identified and dealt with separately. 
    
    Let the path be $P_n=(u_1, u_2, \dots u_n)$, numbered from left to right for $n > 2n_0$. 
    Let $\preprocessinf'(n, i)$ be the output given to $u_i$ (i.e., the vertex at distance $i-1$ from the beginning of $P_n$). Let $\tilde{P}_i$ be a subpath of $P_n$ of length $n_0$, that contains $u_i$ and as many vertices as possible that are at a distance at most $n_0/2$ or fewer. More formally suppose $u_i$ is in the left half of the path, then construct $\tilde{P}_i$ by taking $\min(n_0/2, i)$ vertices to the left of $u_i$ and as many vertices to the right of $u_i$ until its length becomes $n_0$.
    
    \textbf{Preprocessing Stage.} 
    Define $\preprocessinf'(n, u_i) = \preprocessinf(a, b) || \{i \pmod {n_0}\}$, where $||$ denotes concatenation.
    That is, in addition to $\preprocessinf(a, b)$ we also precompute the distance of $u_i$ from the left end of the path (modulo $n_0$). Observe that $i \pmod {n_0}$ provides a labelling of the vertices such that no two nodes within a distance of $n_0$ have the same label.
    
    \textbf{Online Stage.} Label the vertices of the path such that no label repeats within a distance $n_0$. This labelling was already computed in the preprocessing phase, by computing the distance of the vertices modulo $n_0$. Define the modified online algorithm as $\mathcal{A}'(n) = \mathcal{A}(n_0)$, i.e. vertices $u_i$ simply execute $\mathcal{A}$ pretending they are in the path $\tilde{P}_i$ and have labels $i \pmod{n_0}$.
    
    To verify correctness, recall that $4T(n_0) < n_0$, and therefore, when a vertex $u$ executes $\mathcal{A'}$ and assumes a label $\sigma$, it has enough information within its locality to determine label $\sigma'$ of every node within a distance of $T(n_0)$, as all these vertices communicate only with vertices at a distance of $2T(n_0)$ from $u$. Suppose $v$ is at distance at most $T(n_0)$ from $u$. The inputs that $v$ sees within its neighborhood in the path $\tilde{P}_v$ are identical to the inputs that it sees in the path $\tilde{P}_u$, and so the two executions must be the same. Since $\tilde{P}_u$ is labelled correctly for all $u$, the path $P_n$ must also be correctly labelled. 
    
    Finally note that the round complexity of $\mathcal{A}'$ is $T'(n)=$ $T(n_0)$, which is clearly constant, and the total space used per vertex is also a function of only 
    $n_0=O(1)$,
    so $\SUPSPACE(\Pi, P_n) = O(1)$.
\end{proof}

The proof of the above theorem is almost the same as that of Theorem 6 in \cite{chang2019exponential} (for the $\LOCAL$ model) and Theorem 3 in \cite{foerster2019power} (for the $\SUPPORTED$ model). Note that the above theorem is stronger than Theorem 3 of \cite{foerster2019power}, which only translates $o(n)$ time algorithms in $\LOCAL$ to $O(1)$ time algorithms in $\SUPPORTED$, 
whereas our argument also translates $o(n)$ time algorithms in $\SUPPORTED$ to $O(1)$ time algorithms in $\SUPPORTED$. 

\section{Maximal Matching and Maximal Independent Set}

In this section we explore some results on the sub-graph maximal matching and sub-graph maximal independent set problems. By sub-graph maximal matching, we mean a recurrent problem wherein the online instances are simply sub-graphs of the original graph. Similarly for sub-graph MIS. The removed edges can still be used for communication.

\subsection{Maximal Matching}

We look at trees and bounded arboricity graphs. Balliu et al. \cite{Balliu+focs19} showed that in the $\LOCAL$ model, computing a maximal matching for $\Delta$ regular trees deterministically (or with probability at least $1 - \frac{1}{\Delta^\Delta}$) requires $\Omega(\Delta)$ rounds whenever $\Delta >> \log n$. This lower bound trivially extends to the sub-graph maximal matching problem in the $\LOCAL$ model.

In this section we show that for bounded arboricity graphs, maximal matching can be solved in $O(a)$ rounds where $a$ is the arboricity. In particular for trees, this gives an $O(1)$ round algorithm. We note that this already indicates some separation for the general sub-graph maximal matching problem.

\begin{theorem} \label{mm:trees}
    Sub-graph Maximal Matching on trees can be done in $O(1)$ rounds in the $\SUPPORTED$ model.
\end{theorem}
\begin{proof}  
    \textbf{Preprocessing Stage.} Root the tree at an arbitrary vertex and calculate depth of every vertex modulo $2$, denote this by $\depth(v)$ for vertex $v$.
    
    \textbf{Online Stage.} We deal with vertices $v$ that have the same $\depth(v)$ simultaneously. During phase $j$ (for $j = 0, 1$), vertices $v$ with $\depth(v) = j$ and that are not yet matched, send a request to their parent (if there exists an edge from $v$ to its parent in the online stage). Among every vertex $p$ with $\depth(p) = j - 1 \pmod{2}$ that has received a request, $p$ chooses an arbitrary (perhaps one with smallest id) and adds the edge to the maximal matching and removes all incident edges from the graph. When a node $v$ receives an acceptance, it removes all incident edges from the graph (for the next phase). 

    \textbf{Correctness.} When an edge is added to the matching (say $(v, p)$), let $u$ be an arbitrary child of $v$, and $q$ be the parent of $p$.
    Because $\depth(p) \neq \depth(v)$, $(p, w)$ will not be added to the matching in this round. Similarly the edge $(u, v)$ will not be added. So all edges added to the matching in a phase do not violate the constraints for a valid matching.

    At the end of two phases, consider an edge $(v, p)$ that is not part of the matching, where $p$ is the parent of $v$. The reason that edge $(v, p)$ did not get added to the matching was that when $v$ sent a request to $p$ to add an edge, either $p$ accepted another request or $p$ was already matched. In either case, there exists an edge incident on $p$ that is part of the matching, and hence $(v, p)$ cannot be added to the matching.

    \textbf{Round Complexity.} Each phase takes $3$ rounds (one to send the request, one to receive its acceptance into the matching and one to notify neighbors and remove edges from the graph) and there are only $2$ phases.
\end{proof}

\begin{theorem}\label{mm:arb}
    Sub-graph Maximal Matching on graphs of arboricity $a$ can be computed in $O(a)$ LOCAL rounds in the SUPPORTED model.
\end{theorem}

\begin{proof}
    We use the algorithm of Theorem \ref{mm:trees}.
    
    \textbf{Preprocessing Stage.} Compute a forest decomposition of the edge set of the graph. Invoke preprocessing stage of proof of Theorem \ref{mm:trees}.

    \textbf{Remark.} Computing optimal forest decomposition is NP-hard, however one can obtain a forest decomposition into $2a$ forests in polynomial time, so the preprocessing can be made efficient at the expense of only a constant factor in round complexity.

    \textbf{Online Stage.} We compute maximal matching of the subgraph induced by each forest simultaneously. In phase $j$ (for $j = 1, 2, \dots a$), execute sub-graph maximal matching for the $j^{th}$ forest in the decomposition, using the online stage of proof of Theorem \ref{mm:trees}. Before proceeding to the next phase, for every vertex $v$ which got a new incident edge added to the matching, remove all its incident edges from the graph.
\end{proof}

\subsection{Bipartite graphs are the hardest instances to breach the \texorpdfstring{$\Delta$}{Delta} barrier}

Towards finding optimal algorithms for sub-graph Maximal Matching, we show the following lemma.

\begin{lemma} 
    If subgraph MM can be solved in $o(\Delta)$ for all bipartite graphs, then subgraph-MM can be solved in $o(\Delta)$ for all graphs.
\end{lemma}
\begin{proof}
    Suppose subgraph-MM for bipartite graphs can be solved in $O(\frac{\Delta}{f(\Delta)})$ rounds for some function $f(\Delta) $ that is $\omega(1)$.
    Note that MM can be solved in $O(\Delta)$ rounds in the \SUPPORTED model (compute $O(\Delta)$ edge coloring and then compute MIS using said coloring as per Theorem \ref{thm:mis:col}). We make use of the following well known proposition.
    \begin{proposition} 
        For every integer $d$ with $1 \leq d \leq \Delta$, there exists a partition of the vertex set of $G$ into $d$ disjoint sets, $V_1, V_2, \dots V_d$ such that $\Delta(G[V_i]) \leq \Delta(G) / d$.
    \end{proposition}

    \def\deg{\text{deg}}
    \begin{proof}
        Let $d$ be an integer with $1\leq d \leq \Delta$. For a partition $P$ of the vertex set into $V_1, V_2, \dots V_d$, let $\deg_P(v)$ be the degree of the vertex $v$ in $G[V_i]$ where $v \in V_i$. Consider a partition $P$ that minimizes $\sum_{v \in G} \deg_P(v)$. We claim that $P$ is the desired partition. Suppose for contradiction that there exists $v \in G$ with $\deg_P(v) > \Delta / k$. As total degree is at most $\Delta$, there exists $j$ such that $V_j$ contains strictly less than $\Delta / k$ neighbors of $v$. Moving $v$ from $V_i$ to $V_j$, results in reducing $\sum\limits_{v \in G}\deg_P(v)$ by at least $2$ which contradicts the choice of $P$.
    \end{proof}

    \textbf{Preprocessing Stage.} Compute the partition of $V(G)$ into $V_1, V_2, \dots V_d$ for some $2 \leq d \leq \Delta$. (For now the algorithm is parameterised by $d$, we shall find a suitable value for it later).

    \textbf{Online Stage.} First we compute in parallel a MM for each $G[V_i]$, by the linear in $\Delta$ algorithm described earlier. Subsequently matched vertices are removed. We then execute $d$ phases. In phase $i$, consider the bipartite graph induced by the cut $V_i, V \setminus V_i$. Compute MM of this cut after removing the matched vertices in time $O\left(\frac{\Delta}{f(\Delta)}\right)$ and then remove the matched vertices.

    The total round complexity is given by,
    \begin{equation*}  \begin{split}
        T(\Delta) \leq \frac{\Delta}{d} + d \frac{\Delta}{f(\Delta)}
    \end{split}
    \end{equation*}

    Choosing $d = \sqrt{f(\Delta)}$ we get $T(\Delta) \leq \frac{\Delta}{\sqrt{f(\Delta)}}$ which is $o(\Delta)$.
\end{proof}

We propose that the above lemma presents a strong case to study the problem on bipartite graphs. 

\subsection{Maximal Independent Set}

In this section we show that sub-graph maximal matching can be solved in $O(\chi(G))$ rounds in $\SUPPORTED$, where $\chi(G)$ is the chromatic number of the graph $G$. We also show how to extend the underlying idea behind this algorithm for other graph families with ``nice'' vertex decompositions. This extension partly explains the working of the maximal matching algorithm for bounded arboricity graphs (Theorem \ref{mm:trees}).

\begin{theorem}\label{thm:mis:col}
    Sub-graph MIS for a graph $G$ with chromatic number $\chi$ can be solved in $\chi$ rounds.
\end{theorem}
\begin{proof}
    \textbf{Preprocessing Stage.} Compute a proper $\chi$ coloring of the graph.

    \textbf{Online Stage.} In phase $i$ (for $i = 1, 2, \dots \chi$), vertices with color $i$ add themselves to the independent set if none of its neighbors have been added. Once a vertex adds itself to the independent set, it informs all of its neighbors.
\end{proof}

A well known technique of finding a maximal matching (MM) of a graph $G$ is to find a maximal independent set (MIS) of its line graph $L(G)$. Since $\Delta(L(G)) \leq 2 \Delta(G) - 2$, a $o(\Delta)$ algorithm for MIS implies a $o(\Delta)$ algorithm for MM.

There exists an $O(\Delta + \log^* n)$ algorithm for MM in general graphs (after coloring the graph using $\Delta+1$ colors, use Theorem \ref{thm:mis:col}, note that $\Delta + 1$ coloring (i.e. the preprocessing stage) can be done in $o(\Delta) + \log^* n$ rounds).

\begin{lemma}\label{lem:mis:split}
    Suppose sub-graph MIS can be solved for a graph family $\mathcal{F}$ in $T$ rounds, and suppose the vertex set of a graph $G$ can be covered by $k$ sets $V_1, V_2, \dots V_k$ (i.e. $\cup V_i = V(G)$) such that $G[V_k] \in \mathcal{F}$, then sub-graph MIS for $G$ can be solved in $T \cdot k + (k - 1)$ rounds in the SUPPORTED model.
\end{lemma}

\begin{proof}
    The above lemma is in some sense a generalization of Theorem \ref{thm:mis:col}. Applying $\mathcal{F}$ to be the set of all graphs with $0$ edges and observing that MIS can be solved trivially in $T = 0$ rounds of communication. Chromatic number $\chi$ implies that its vertex set can be partitioned into $\chi$ sets each of which is independent. Moving on to the proof (which is in spirit same as that of Theorem \ref{thm:mis:col})

    \textbf{Preprocessing Stage.} Decompose $G$ into $V_1, V_2, \dots V_k$

    \textbf{Online Stage.} The algorithm runs in several phases. In each phase the independent set computed so far is extended. Let $I_i$ be the independent set computing during phase $i$ ($I_0 = \emptyset$). In phase $i$, compute the MIS for the subgraph $G[V_i \setminus I_{i-1}]$. Note that $G[V_i \setminus I_{i-1}]$ is a subgraph of $G[V_i]$ and by the lemma, MIS can be computed in $T$ rounds. 

    Let $\tilde{I}_i$ be the independent set of $G[V_i \setminus I_{i-1}]$. Observe that $I_i = I_{i-1} \cup \tilde{I}_i$ is an independent set of $G$. 

    Finally consider $I_k$. Suppose $I_k$ was not an MIS, then there exists a vertex $v \not\in I_k$ such that none of the neighbors of $v$ lie in $I_k$. Suppose $v \in V_i$, then consider $H = G[V_i \setminus I_{i-1}]$. We know $v \not\in I_{i-1}$ so $v \in H$. Also $v \not\in \tilde{I}_i$, the MIS of $H$. Consider $J = \tilde{I}_i \cup \{v\}$. $J$ is an independent set of $H$, since none of the neighbors of $v$ lie in $I$. This contradicts that $\tilde{I}_i$ is an MIS of $H$.
\end{proof}

One can also view the results for MM (Theorems \ref{mm:trees}, \ref{mm:arb}) as a consequence of the above lemma. MM for a tree (say $T$) is identical to the MIS for its line graph. Observe that MIS is trivial for graphs $G$ where each of connected component of $G$ is a clique (vertex $v$ belongs to MIS iff it has least ID among its neighbors). Let this family of graphs be $\mathcal{F}$. Root the tree $T$ at an arbitrary vertex and let $E_0, E_1$ be the set of edges at even and odd depths respectively. The line graph induced by $E_0$ (i.e. $L(T)[E_0]$) belongs to $\mathcal{F}$. Similarly for bounded arboricity graphs, we can decompose them into $a$ forests and in turn into $2a$ members of $\mathcal{F}$.

\section*{Acknowledgement}
David Peleg holds the Venky Harinarayanan and Anand Rajaraman (VHAR) Visiting Chair
Professorship at IIT Madras. This work was carried out in part during mutual visits that were supported by the VHAR Visiting Chair funds. John Augustine and Srikkanth Ramachandran are supported by the Centre of Excellence in Cryptography, Cybersecurity, and Distributed Trust (CCD) and by an IITM-Accenture project (SB/22-23/007/JOHN/ACC). 

\bibliography{main.bib}

\end{document}